\documentclass[10pt, twocolumn, comsoc]{IEEEtran}

\usepackage{graphicx,epsfig}
\usepackage[noadjust]{cite}
\usepackage{mcite}
\usepackage{amsfonts,helvet}
\usepackage{fancyhdr}
\usepackage{threeparttable}
\usepackage{epsf,epsfig}
\usepackage{amsthm}
\usepackage{amsmath}
\usepackage{siunitx}
\usepackage{amssymb}
\usepackage{dsfont}
\usepackage{subfigure}
\usepackage{color}
\usepackage[linesnumbered,ruled]{algorithm2e}
\usepackage{enumerate}
\usepackage{cancel}
\usepackage{comment}
\usepackage{siunitx}

\usepackage[colorlinks=true, linkcolor=blue]{hyperref}

\newcommand{\fig}[1]{Fig.\ \ref{#1}}
\newtheorem{theorem}{Theorem}

\newtheorem{corollary}{Corollary}

\newtheorem{lemma}{Lemma}

\newtheorem{remark}{Remark}

\SetKwInput{KwInput}{Initialize}
\SetKwInput{KwOutput}{Output}
\SetKwProg{Fn}{Stage 1}{}{}
\SetKwProg{Fn}{Stage 2}{}{}

\usepackage{eucal}

\begin{document}

\title{Space-Time Adaptive Beamforming for Satellite Communications: Harnessing Doppler as New Signaling Dimensions}

\author{Hyeongtak~Yun, Seyong~Kim, and Jeonghun~Park

\thanks{
This work was supported by Samsung Research Funding \& Incubation Center of Samsung Electronics under Project Number SRFC-IT2402-06.
H. Yun, S. Kim, and J. Park are with the School of Electrical and Electronic Engineering, Yonsei University, Seoul 03722, South Korea (e-mail: {\texttt{yht3114@yonsei.ac.kr; sykim@yonsei.ac.kr; jhpark@yonsei.ac.kr}}).
}}

\maketitle \setcounter{page}{1}

\begin{abstract}
Low Earth orbit (LEO) satellite downlinks are fundamentally limited by severe channel correlation: the line-of-sight (LoS)-dominant propagation and high orbital altitude confine users to a narrow angular region, rendering the multiuser channel matrix ill-conditioned. 
This paper provides a rigorous characterization of this limitation by exploiting the Vandermonde structure of the channel. Specifically, we link the minimum eigenvalue of the channel Gram matrix to user crowding through a balls-and-bins abstraction, and derive asymptotic sum rate scaling laws for both uniform linear arrays and uniform planar arrays. Our analysis reveals a sharp density threshold beyond which zero-forcing (ZF) precoding provably fails. To overcome this spatial multiplexing breakdown, we propose space-time adaptive beamforming (STAB), which exploits user-dependent residual Doppler shifts as an additional discrimination dimension. By constructing a time-extended channel in the joint space-Doppler domain, STAB restores a non-vanishing sum rate in regimes where purely spatial ZF collapses, {with a matching achievability bound.} We further develop a space-Doppler user selection (SDS) algorithm that leverages both spatial and Doppler separability for scheduling. Numerical results corroborate the analytical predictions and demonstrate that the combination of STAB and SDS achieves substantial sum rate gains over conventional methods in realistic dense LEO downlink scenarios.
\end{abstract}


\section{Introduction}
\label{sec:introduction}
Low Earth orbit (LEO) satellite communications are increasingly regarded as a viable complement to terrestrial networks for extending broadband connectivity to underserved regions. Driven by significant reductions in launch costs and advances in small-satellite platforms, commercial mega-constellations such as Starlink, OneWeb, and Kuiper are being deployed at scales of thousands of satellites, targeting near-global coverage. While initial LEO deployments have been primarily motivated by coverage expansion, the next phase of system evolution is shaped by capacity and throughput requirements \cite{satelliteprecoding:wcm:16}. Under limited onboard spectrum and power budgets, improving spectral efficiency in the satellite downlink has emerged as a central design objective \cite{slim:survey:24}. In particular, serving a large number of users over the same time-frequency resources demands spatial multiplexing capabilities that go beyond conventional orthogonal multiple-access strategies. This motivates the adoption of multiuser multiple-input multiple-output (MU-MIMO) precoding techniques on the satellite downlink, which can substantially enhance system throughput by exploiting the spatial domain to simultaneously serve multiple co-channel users \cite{chris:twc:15, zheng:twc:2012, vazquez:jsac:18, seyong:twc:25}.
\begin{figure}[!t]
    \centerline{\resizebox{0.5\columnwidth}{!}{\includegraphics{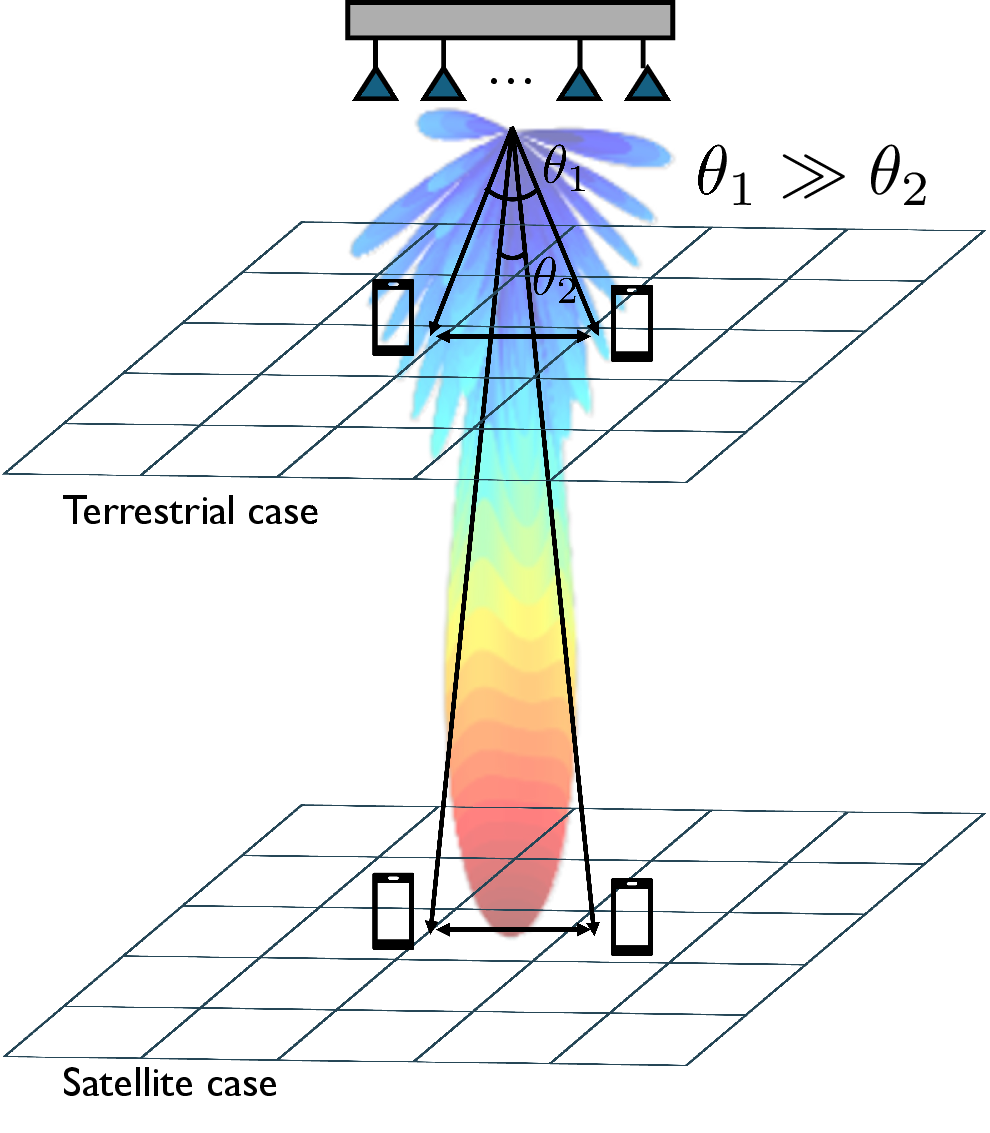}}}
    \caption{Large orbital altitudes inducing user channel correlation and MU-MIMO limitations.}
    \label{fig:mu_mimo_limit}
\end{figure}

Despite the potential of MU-MIMO, its effectiveness in satellite downlinks is fundamentally limited by the propagation environment \cite{angeletti:access:20}. Unlike terrestrial cellular systems, where rich multipath scattering provides high-rank channel matrices that facilitate spatial multiplexing, satellite user links are dominated by line-of-sight (LoS) propagation with very limited scattering \cite{you:jsac:20}. In this regime, the spatial separability of users is governed primarily by the angular separation as seen from the satellite array. Because LEO satellites operate at orbital altitudes on the order of hundreds of kilometers---far exceeding typical terrestrial cell radii---users within the same coverage area subtend only a very narrow range of angles at the satellite, resulting in highly correlated channel vectors \cite{jsdm:tit:13} (as depicted in Fig.~\ref{fig:mu_mimo_limit}). Consequently, the multiuser channel matrix tends to be severely ill-conditioned, and the number of users that can be effectively multiplexed over the same time-frequency resource is much smaller than what the antenna count would nominally permit.

Nonetheless, this fundamental limitation is not yet well understood from an analytical perspective. Most existing works on satellite MU-MIMO have assessed precoding performance through numerical simulations \cite{zheng:twc:2012, you:jsac:20, chris:twc:15}, rather than leveraging the mathematical structure of LoS-dominant channels. In particular, when users are served via linear precoding such as zero-forcing (ZF), the multiuser channel matrix exhibits a Vandermonde-like structure whose conditioning is intimately tied to the spatial geometry of the users. Yet, a rigorous characterization of how this structure governs the achievable MU-MIMO performance is still lacking in the existing literature, especially in the regime where the angular separations among users become small relative to the array resolution. This gap makes it difficult to establish analytical performance limits and principled design guidelines for LEO MU-MIMO downlink systems.

In this paper, we address these challenges through two main contributions. First, we perform a rigorous performance analysis of ZF precoding for the satellite downlink MU-MIMO system by exploiting the Vandermonde structure of the LoS channel matrix. To be specific, we characterize the achievable multiplexing performance as a function of cell size, user count, and antenna configuration in the asymptotic antenna regime. Second, motivated by the fundamental limitations revealed by our analysis, we propose a space-time adaptive beamforming (STAB) framework, inspired by classical space-time adaptive processing (STAP) in airborne radar \cite{melvin:maes:04,wicks:msp:06}. STAB exploits the fact that closely spaced users, despite having nearly identical spatial signatures, exhibit distinct Doppler shifts arising from their individual velocity vectors. By doing this, the proposed STAB provides a simple yet powerful additional dimension for multiuser separation beyond what purely spatial precoding can offer.

\subsection{Related Works}
\label{subsec:related_works}
Recently, satellite communications have increasingly adopted aggressive full frequency reuse, spurring extensive research on precoding techniques for inter-beam interference management in multibeam systems \cite{satelliteprecoding:wcm:16, slim:survey:24, mysore:commlett:21}. In this context, multicast multigroup precoding with user scheduling \cite{chris:twc:15} and sum rate maximization under quality-of-service constraints \cite{li:tcom:23} have been studied. Under practical channel state information (CSI) limitations in satellite links, ZF precoding with partial CSI \cite{ahmad:tvt:21}, precoding under outdated CSI \cite{vazquez:jsac:18}, and LEO downlink precoders exploiting statistical CSI \cite{you:jsac:20} have been explored, while \cite{angeletti:access:20} proposed a pragmatic massive MIMO architecture for broadband satellites. Recently, distributed precoding for satellite-terrestrial integrated networks \cite{doseon:twc:25} and asymptotic performance analyses of multibeam massive MIMO systems \cite{seyong:wcl:25, seyong:twc:25} have further advanced the field.

As explained earlier, a key feature of the satellite channel is the dominance of LoS propagation combined with the difficulty of achieving sufficient angular separation among users, which makes user channel vectors highly correlated and limits spatial multiplexing. This issue has been extensively studied in the massive MIMO literature, including favorable propagation analysis \cite{ngo:eusipco:14}, max-min power control under LoS correlation \cite{yang:tcom:17}, and measurement-based evidence of limited spatial separability \cite{flordelis:access:18}. These observations have motivated correlation-aware user selection \cite{chaves:ojcoms:22} and spatial-domain schemes tailored to LEO, such as space-angle user grouping \cite{you:jsac:20}, planar array design exploiting channel geometry \cite{kexinli:tcom:22}, and graph-based user clustering for LEO MU-MIMO \cite{riviello:asmsspsc:22, ahmad:eucnc:23}. However, these studies mainly focus on practical algorithms rather than providing a structural understanding of why spatial-domain MU-MIMO fails in dense LEO geometries. It also remains analytically unclear under what conditions spatial multiplexing fundamentally breaks down and how this limit scales with key system parameters.

It is well established that under LoS propagation, the MU-MIMO channel matrix exhibits a Vandermonde structure, making the spectral analysis of Vandermonde matrices essential to characterizing LEO satellite MU-MIMO performance. The spectral properties of such matrices have been studied from two complementary perspectives: random matrix theory, including asymptotic moment analysis \cite{ryan:tit:09} and eigenvalue distribution bounds \cite{tucci:tit:11, tucci:jtp:14}; and super-resolution theory, where sharp phase transitions in the condition number \cite{moitra:stoc:15} and tight bounds linking the smallest singular value to cluster size \cite{batenkov:siam:20, li:stable:2021} have been established. These tools have also been applied to classical subspace methods \cite{li:tit:20, liu:tit:21}. Despite this extensive foundation, a formal bridge between these spectral characterizations and the fundamental limits of LEO MU-MIMO systems has yet to be established.
While existing satellite MU-MIMO studies have predominantly relied on spatial degrees of freedom, the radar domain offers a fundamentally different approach: airborne STAP separates targets sharing identical spatial angles via slow-time Doppler filtering \cite{melvin:maes:04, wicks:msp:06}, and synthetic aperture radar synthesizes virtual apertures through platform motion \cite{moreira:mgrs:13}, both demonstrating that temporal evolution can provide additional discrimination when spatial resolution is insufficient. Recently, \cite{yim:twc:26} proposed a STAP-inspired space-time beamforming framework for multi-LEO satellite systems exploiting angle of arrival (AoA) and relative Doppler to synthesize a larger virtual aperture. Nevertheless, this work focuses on precoder design for interference mitigation without providing a fundamental analysis of the performance collapse caused by spatial separability loss.

\subsection{Contributions}
\label{subsec:contributions}
The main contributions of this paper are summarized as follows:
\begin{itemize}
    \item By exploiting the Vandermonde structure of LoS-dominant LEO satellite channels, we model geometric user crowding through a balls-and-bins abstraction. This allows us to derive asymptotic sum rate scaling laws and characterize the user density threshold beyond which purely spatial multiplexing fundamentally breaks down.
    \item To overcome the aforementioned spatial multiplexing limits, we propose STAB, a framework leveraging user-dependent residual Doppler shifts as an extra signaling dimension. By repeating symbol transmissions, STAB expands the conventional spatial channel into a joint space-Doppler domain to enhance user separability.
    \item We extend the asymptotic scaling analysis to STAB, proving its ability to maintain a non-vanishing sum rate where spatial ZF fails. Inspired by these insights, we further develop a space-Doppler user selection (SDS) algorithm that achieves substantial sum rate improvements over conventional methods in {realistic LEO downlink systems}.
\end{itemize}

\subsection{Notation}
\label{subsec:notation}
Scalars, vectors, and matrices are denoted by $a$, $\mathbf{a}$, and $\mathbf{A}$, respectively. $\mathbb{C}^{m\times n}$ is the set of $m\times n$ complex matrices, and $\mathbb{N}$ is the set of positive integers. $(\cdot)^{\mathsf{T}}$, $(\cdot)^{\mathsf{H}}$, and $(\cdot)^{-1}$ denote transpose, Hermitian transpose, and inverse, respectively. $\operatorname{tr}(\cdot)$ denotes the trace, $\lambda_1(\mathbf{A})$ the minimum eigenvalue of a Hermitian matrix $\mathbf{A}$, and $|\cdot|$ the magnitude or set measure, as appropriate. $\lfloor\cdot\rfloor$ and $\lceil\cdot\rceil$ denote the floor and ceiling. $\mathbb{P}(\cdot)$, $\mathbb{E}[\cdot]$, and $\mathbf{1}\{\cdot\}$ denote probability, expectation, and the indicator function, respectively. $\mathbf{I}_N$ is the $N\times N$ identity matrix. $\|\cdot\|_2$ denotes the Euclidean norm for vectors and the spectral norm for matrices, $\|\cdot\|_1$ and $\|\cdot\|_\infty$ the standard vector norms or induced matrix norms, as appropriate, and $\|\cdot\|_F$ the Frobenius norm. $\otimes$ denotes the Kronecker product. $\mathcal{CN}(\cdot,\cdot)$ denotes the circularly symmetric complex Gaussian distribution. For nonnegative $f$ and $g$, $f=O(g)$ means that there exist constants $C>0$ and $M_0$ such that $f\le Cg$ for all $M\ge M_0$, $f=\Omega(g)$ means that $g=O(f)$, $f=\Theta(g)$ means that both hold, and $f=o(g)$ means that $f/g\to0$ as $M\to\infty$. We write $f\lesssim g$ when $f\le Cg$ holds for a constant $C>0$ independent of $M$.

\section{System Model}
\label{sec:system_model}

\subsection{System Setup}
\label{subsec:system_setup}
We consider LEO satellite downlink systems equipped with phased-array antennas, which enable rapid electronic beam steering. Phased-array antennas have become the baseline architecture in recent LEO communications literature owing to advances in hardware and digital beamforming \cite{heo:sate:2023, you:hybrid:2022, you:massive:2022, kexinli:tcom:22}. A satellite at an altitude $H$ serves $K$ single-antenna ground users. The satellite is equipped with a uniform planar array (UPA) of $M = M_x M_y$ antenna elements on the $xy$-plane, with half-wavelength spacing. The $K$ users are independently and identically distributed (i.i.d.) uniformly over a square service area $\CMcal{A}=[-R,R]\times[-R,R]$, where $R$ is the cell size and $(x_k,y_k)$ denotes the horizontal coordinates of user $k$.

We assume the satellite possesses full CSI. This assumption is well-justified in the LoS-dominant regime, where the channel is largely determined by the geometric relationship between the satellite and each user. Since user terminals in 3GPP non-terrestrial networks (NTNs) are required to be equipped with Global Navigation Satellite System (GNSS) receivers for timing and frequency pre-compensation \cite{3gpp_tr38821}, accurate position and velocity information is readily available and can be reported to the satellite with minimal overhead.

\subsection{LEO Downlink Channel Model}
\label{subsec:channel_model}
This section describes the LEO downlink channel model. For UPA, the steering vector $\mathbf{a}_k \in \mathbb{C}^{M\times 1}$ is defined as
\begin{align}
\mathbf{a}_k = \mathbf{a}(u_{y, k}; M_y) \otimes \mathbf{a}(u_{x, k}; M_x),
\label{eq:upa_a_theta_phi}
\end{align}
where $u_{x, k} = \frac{\sin\theta_k \cos\phi_k}{2}$ and $u_{y, k} = \frac{\sin\theta_k \sin\phi_k}{2}$ are the spatial frequencies, with $\theta_k$ and $\phi_k$ being the zenith and azimuth angles of user $k$ relative to the satellite. The response vector $\mathbf{a}(x; M) \in \mathbb{C}^{M}$ is given by
\begin{align}
\mathbf{a}(x;M) = \frac{1}{\sqrt{M}} \left[1,\ e^{j2\pi x},\ \ldots,\ e^{j2\pi (M-1)x}\right]^{\mathsf T}.
\label{eq:array_response_vector}
\end{align}

The LEO downlink channel for user $k$ is a superposition of $P_k$ multipath components:
\begin{align}
\mathbf{h}_k(t,f) = \sqrt{M} \sum_{p=0}^{P_k-1} \beta_{k,p} e^{j2\pi (t \nu_{k,p} - f \tau_{k,p})} \mathbf{a}_k,
\label{eq:leo_general_tf_theta_phi}
\end{align}
where $\beta_{k,p}$, $\tau_{k,p}$, and $\nu_{k,p}$ denote the complex gain, delay, and Doppler shift of path $p$ for user $k$, respectively.

Due to the large propagation distance and localized scatterers near the user, the satellite-induced Doppler $\nu_k^{\mathrm{sat}}$ and minimum delay $\tau_{k,0}$ are common across all paths. The channel can thus be factorized as
\begin{align}
\mathbf{h}_k(t,f) = e^{j2\pi\left(t\nu_k^{\mathrm{sat}}-f\tau_{k,0}\right)} \, \beta_k(t,f)\, \sqrt{M}\,\mathbf{a}_k,
\label{eq:leo_factored_tf_theta_phi}
\end{align}
where $\beta_k(t,f)$ is a Rician fading coefficient with Rician factor $\kappa_k$. Here, $\nu_{k,p}^{\mathrm{ut}} = \nu_{k,p} - \nu_k^{\mathrm{sat}}$ and $\tau_{k,p}^{\mathrm{ut}} = \tau_{k,p} - \tau_{k,0}$ represent the user-side relative Doppler and normalized delay. After compensating for the common phase term via synchronization, the effective baseband channel reduces to
\begin{align}
\mathbf{h}_k(t,f) = \sqrt{M}\,\beta_k(t,f)\,\mathbf{a}_k.
\label{eq:leo_baseband_theta_phi}
\end{align}

We adopt a single-path model for \eqref{eq:leo_baseband_theta_phi}, motivated by the LoS-dominant nature of LEO propagation, particularly with high elevation angles and directive beamforming. In such scenarios, the Rician factor $\kappa_k$ is typically large enough to render scattered components negligible. Retaining only the dominant LoS component, i.e., $P_k = 1$ for all $k$, the channel simplifies to
\begin{align}
\mathbf{h}_k(t) = \sqrt{M}\, \beta_k e^{j2\pi \nu_k^{\mathrm{ut}} t} \mathbf{a}_k,
\label{eq:single_path_channel}
\end{align}
where $\beta_k = \beta_{k,0}$ is the complex gain of the LoS path.

\subsection{LEO Downlink Signal Model}
\label{subsec:signal_model}
In this section, we present the multiuser LEO downlink signal model. Although user mobility introduces a residual Doppler shift $\nu_k^{\mathrm{ut}}$, we follow the quasi-static assumption common in satellite communications, where this effect is negligible during a single transmission interval. Accordingly, we initially set $\nu_k^{\mathrm{ut}} = 0$; the impact of Doppler variations is addressed in Section~\ref{sec:stab}.

The satellite serves $K$ users simultaneously using a precoded signal vector $\mathbf{x} \in \mathbb{C}^M$:
\begin{align}
\mathbf{x} = \sqrt{P}\sum_{k=1}^{K} \mathbf{f}_k s_k = \sqrt{P}\mathbf{F} \mathbf{s},
\label{eq:tx_signal}
\end{align}
where $P$ is the transmit power, $\mathbf{s} = [s_1, \dots, s_K]^{\mathsf{T}} \in \mathbb{C}^{K}$ is the data symbol vector with $\mathbf{s} \sim \mathcal{CN}(\mathbf{0}, \mathbf{I}_K)$, and $\mathbf{F} = [\mathbf{f}_1, \dots, \mathbf{f}_K] \in \mathbb{C}^{M \times K}$ is the precoding matrix satisfying $\|\mathbf{F}\|_F^2 = 1$. The received signal $y_k$ at user $k$ is
\begin{align}
y_k &= \mathbf{h}_k^{\mathsf{H}} \mathbf{x} + z_k = \sqrt{P}\,\mathbf{h}_k^{\mathsf{H}} \mathbf{f}_k s_k + \sqrt{P}\sum_{i \neq k} \mathbf{h}_k^{\mathsf{H}} \mathbf{f}_i s_i + z_k,
\label{eq:received_signal_k}
\end{align}
where $\mathbf{h}_k$ is the baseband channel from \eqref{eq:single_path_channel} with $\nu_k^{\mathrm{ut}} = 0$, and $z_k \sim \mathcal{CN}(0, \sigma^2)$ is additive white Gaussian noise (AWGN). In \eqref{eq:received_signal_k}, the first term is the desired signal, and the second represents inter-user interference (IUI). The received signal vector $\mathbf{y} = [y_1, \dots, y_K]^{\mathsf{T}} \in \mathbb{C}^K$ is expressed as
\begin{align}
\mathbf{y} = \sqrt{P}\,\mathbf{H}^{\mathsf{H}} \mathbf{F} \mathbf{s} + \mathbf{z},
\label{eq:received_signal_matrix}
\end{align}
where $\mathbf{H} = [\mathbf{h}_1, \dots, \mathbf{h}_K] \in \mathbb{C}^{M \times K}$ is the channel matrix, and $\mathbf{z} = [z_1, \dots, z_K]^{\mathsf{T}}$ is the noise vector.

\subsection{Sum Rate Characterization with ZF}
\label{subsec:zf_characterization}
We adopt ZF precoding to mitigate IUI. The precoding matrix is $\mathbf{F} = \eta \mathbf{H}(\mathbf{H}^{\mathsf{H}}\mathbf{H})^{-1}$, where $\eta = \sqrt{1 / \operatorname{tr}((\mathbf{H}^{\mathsf{H}}\mathbf{H})^{-1})}$ is the power normalization factor. Under the ZF criterion, IUI is perfectly nullified, and the signal-to-interference-plus-noise ratio (SINR) for user $k$ is
\begin{align}
\mathrm{SINR}_{k}
= \frac{P}{\sigma^2 \operatorname{tr} \left(\left({\mathbf{H}}^{\mathsf H}{\mathbf{H}}\right)^{-1}\right)}
= \frac{\rho M}{{\operatorname{tr}({\mathbf{G}}^{-1})}},
\label{eq:zf_sinr}
\end{align}
where $\rho = P/\sigma^2$ is the transmit signal-to-noise ratio (SNR) and ${\mathbf{G}} = \frac{{\mathbf{H}}^{\mathsf H}{\mathbf{H}}}{M}$ is the channel Gram matrix. We evaluate system performance using the spatially averaged sum rate. For users i.i.d. uniform within $\CMcal{A}$, the average sum rate is expressed as
\begin{align}
\mathbb{E}[R_{\Sigma}]
= K\,\mathbb{E} \Bigg[ \log_{2} \left(1 + \frac{\rho M}{{\operatorname{tr}({\mathbf{G}}^{-1})}}\right)\Bigg],
\label{eq:average_sum_rate}
\end{align}
where the expectation is taken over random user locations.

To gain insight into how user crowding degrades ZF performance in LoS channels, consider a two-user case with $\mathbf{h}_k = \sqrt{M}\mathbf{a}_k$. In this scenario, the channel Gram matrix simplifies to
\begin{align}
\mathbf{G} = \begin{bmatrix} 1 & g \\ g^* & 1 \end{bmatrix},
\label{eq:two_user_gram}
\end{align}
where $g = \mathbf{a}_1^{\mathsf{H}}\mathbf{a}_2$ is the spatial correlation. For $K=2$, $\operatorname{tr}(\mathbf{G}^{-1}) = 2/(1-|g|^2)$, yielding the instantaneous sum rate
\begin{align}
R_\Sigma &= 2\log_2\left(1 + \frac{\rho M}{2}(1-|g|^2)\right). \label{eq:sumrate_2user}
\end{align}
For UPA steering vectors, the correlation magnitude factorizes as $|g| = |g_x||g_y|$ via Dirichlet kernels:
\begin{align}
|g_x| = \left|\frac{\sin(\pi M_x \Delta u_{x})}{M_x\sin(\pi \Delta u_{x})}\right|, \quad |g_y| = \left|\frac{\sin(\pi M_y \Delta u_{y})}{M_y\sin(\pi \Delta u_{y})}\right|,
\label{eq:dirichlet_kernels}
\end{align}
with $\Delta u_{x} = u_{x,1} - u_{x,2}$ and $\Delta u_{y} = u_{y,1} - u_{y,2}$. As users become angularly close, $\Delta u_{x}$ and $\Delta u_{y}$ approach zero, driving $|g|$ toward 1. Consequently, \eqref{eq:sumrate_2user} implies that both the SINR and sum rate collapse toward zero.

\section{Space-Time Adaptive Beamforming}
\label{sec:stab}
This section introduces STAB, which exploits the user-induced residual Doppler shift as an additional degree of freedom for user separation, relaxing the quasi-static assumption ($\nu_k^{\mathrm{ut}}=0$) of Section~\ref{sec:system_model}. By transmitting the same symbol vector over $L$ consecutive snapshots spaced by an interval $T_r$, each user experiences a distinct phase shift across slow-time. This repeated transmission yields a time-extended channel representation, allowing the satellite to jointly discriminate users through their spatial and residual Doppler-induced temporal signatures.

\subsection{Space-Time Channel and Signal Model}
\label{subsec:stab_channel_signal}
We first discretize the channel with a sampling interval $T_r$. For snapshot $\ell \in \{0,1,\dots,L-1\}$, the sampled channel vector of user $k$ is given by
\begin{align}
\mathbf{h}_k[\ell]
= \mathbf{h}_k(\ell T_r)
= \sqrt{M}\beta_ke^{j2\pi \ell \omega_k}\mathbf{a}_k,
\label{eq:stab_snapshot_channel}
\end{align}
where $\omega_k=\nu_k^{\mathrm{ut}}T_r$ denotes the normalized user-induced residual Doppler frequency.
In radar systems \cite{melvin:maes:04,wicks:msp:06}, target motion is approximated as a constant velocity during a coherent processing interval, allowing the target response to be factored into a spatial component and a slow-time Doppler phase shift. Aligning well with this concept, we assume that $\beta_k$ and $\mathbf{a}_k$ remain constant over the observation interval $LT_r$, while the phase shifts across snapshots, providing additional Doppler-domain separability.

\begin{figure}[!t]
\centerline{\resizebox{0.9\columnwidth}{!}{\includegraphics{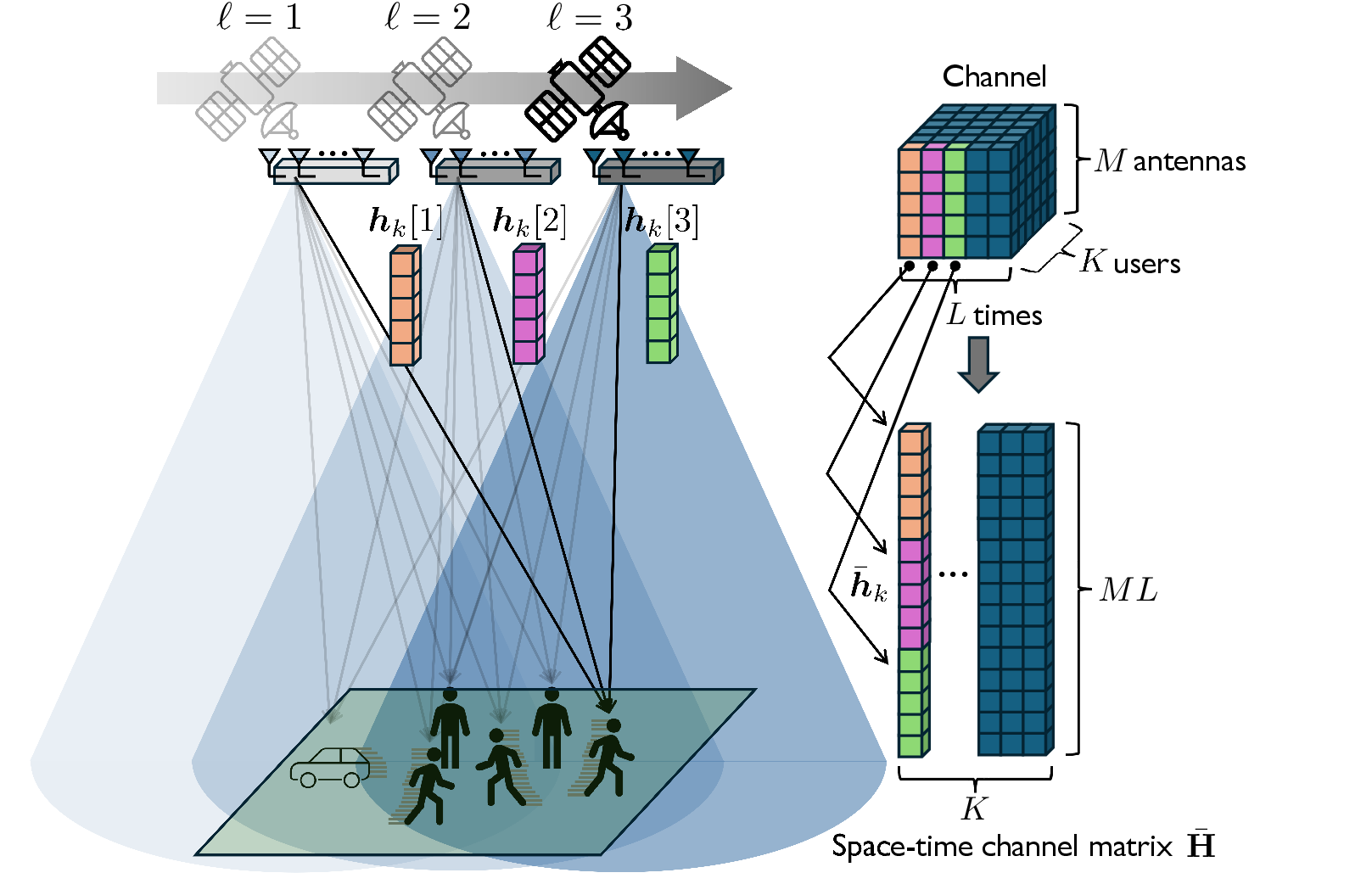}}}
\caption{Considered system model for STAB with repetition length $L=3$.}
\label{fig:system_model}
\end{figure}

To exploit this slow-time structure, the satellite transmits the same symbol vector $\mathbf{s}$ over $L$ consecutive snapshots using a precoder $\mathbf{F}[\ell]$, yielding $\mathbf{x}[\ell]= \sqrt{P}\mathbf{F}[\ell]\mathbf{s}$.
We construct the space-time representation by stacking these $L$ snapshots, as illustrated in Fig.~\ref{fig:system_model}.
Defining the temporal steering vector
\begin{align}
\mathbf{b}_k = \frac{1}{\sqrt{L}}
\begin{bmatrix}
1,e^{j2\pi\omega_k}, \cdots,e^{j2\pi(L-1)\omega_k}
\end{bmatrix}^{\mathsf T},
\end{align}
the stacked channel vector $\bar{\mathbf{h}}_k \in \mathbb{C}^{ML}$ of user $k$ is written as
\begin{align}
\bar{\mathbf{h}}_k = \sqrt{ML}\beta_k(\mathbf{b}_k \otimes \mathbf{a}_k).
\end{align}
Similarly, the stacked space-time precoder $\bar{\mathbf{F}} \in \mathbb{C}^{ML \times K}$ is defined as
\begin{align}
\bar{\mathbf{F}} =
\begin{bmatrix}
\mathbf{F}[0]; \cdots; \mathbf{F}[L-1]
\end{bmatrix}
=
\begin{bmatrix}
\bar{\mathbf{f}}_1, \cdots, \bar{\mathbf{f}}_K
\end{bmatrix},
\label{eq:stab_st_precoder}
\end{align}
where $\bar{\mathbf{f}}_k$ is the space-time precoding vector for user $k$.
By aggregating the normalized observations over $L$ snapshots, the overall space-time received signal $\bar{\mathbf{y}} \in \mathbb{C}^K$ is given by
\begin{align}
\bar{\mathbf{y}} = \sqrt{\frac{P}{L}}\bar{\mathbf{H}}^{\mathsf H}\bar{\mathbf{F}}\mathbf{s} + \bar{\mathbf{z}},
\label{eq:stab_system_model}
\end{align}
where $\bar{\mathbf{H}}=[\bar{\mathbf{h}}_1,\dots,\bar{\mathbf{h}}_K]\in\mathbb{C}^{ML\times K}$ is the space-time channel matrix, and $\bar{\mathbf{z}} \sim \mathcal{CN}(\mathbf{0}, \sigma^2\mathbf{I}_K)$ is the equivalent noise vector.

\subsection{Sum Rate Characterization with STAB}
\label{subsec:stab_sum_rate}
We apply ZF precoding to the space-time channel $\bar{\mathbf{H}}$. Under $\|\bar{\mathbf{F}}\|_F^2 = L$, the precoder is $\bar{\mathbf{F}} = \bar{\eta} \bar{\mathbf{H}}(\bar{\mathbf{H}}^{\mathsf{H}}\bar{\mathbf{H}})^{-1}$ with $\bar{\eta} = \sqrt{L/\operatorname{tr}((\bar{\mathbf{H}}^{\mathsf{H}}\bar{\mathbf{H}})^{-1})}$. {This normalization fixes the average transmit power per snapshot, so STAB shares the power budget of spatial ZF and requires no additional hardware, the only cost growing with $L$ being the baseband precoder processing, which is negligible against the amplifier and payload power.} The resulting SINR for user $k$ is given by
\begin{align}
\overline{\mathrm{SINR}}_k = \frac{\rho}{L}\cdot\frac{L}{\operatorname{tr}\left((\bar{\mathbf{H}}^{\mathsf{H}}\bar{\mathbf{H}})^{-1}\right)} = \frac{\rho ML}{\operatorname{tr}\left(\bar{\mathbf{G}}^{-1}\right)},
\label{eq:stab_sinr}
\end{align}
where $\rho = P/\sigma^2$ and $\bar{\mathbf{G}} = \frac{1}{ML}\bar{\mathbf{H}}^{\mathsf{H}}\bar{\mathbf{H}}$ is the space-time channel Gram matrix. Since $L$ snapshots are consumed to transmit each symbol vector, the average sum rate is formulated as
\begin{align}
\mathbb{E}[\bar{R}_{\Sigma}] = \frac{K}{L} \mathbb{E} \left[ \log_2\left(1 + \frac{\rho ML}{\operatorname{tr}\left(\bar{\mathbf{G}}^{-1}\right)}\right) \right],
\label{eq:stab_sumrate}
\end{align}
where the expectation is taken over the user locations and their associated Doppler shifts.

\begin{remark}[Intuition on STAB]\normalfont
STAB effectively expands the channel dimensionality from $M$ to $ML$ by augmenting the spatial signature $\mathbf{a}_k$ with the temporal signature $\mathbf{b}_k$. This mechanism provides several key insights:

\textbf{Synthetic virtual array:} STAB can be interpreted as forming a synthetic virtual array across the joint spatial and temporal dimensions. By collecting $L$ snapshots, the residual Doppler shifts of each user introduce distinct phase shifts over slow time, effectively expanding the system aperture from $M$ to $ML$. When the residual Doppler shifts are sufficiently resolvable over the observation interval, this Doppler-induced phase diversity plays a role similar to spatial phase diversity in a conventional array, allowing the satellite to resolve users that are nearly inseparable in the spatial domain alone. Consequently, the space-time user channels become less correlated, leading to a better-conditioned channel Gram matrix and improved robustness of ZF precoding. {However, this virtual array expansion does not imply a physical array expansion. Physical expansion enriches the spatial manifold, whereas STAB enriches the temporal manifold. Precisely because the two draw on different resources, with physical expansion requiring additional hardware and temporal expansion incurring the repetition penalty discussed below, they serve as complementary means of enhancing user separability.}

\textbf{Separability-efficiency tradeoff:} STAB introduces a fundamental tradeoff between improved user separability and rate. A larger $L$ provides more degrees of freedom to decorrelate user channels and reduce the noise enhancement of ZF precoding, but it also incurs a $1/L$ pre-log penalty because each symbol vector is repeated over $L$ snapshots. Hence, STAB is particularly effective in interference- and separability-limited regimes, where the separability gain from temporal expansion dominates the associated loss in transmission efficiency.
\end{remark}

\section{Asymptotic Analysis of Satellite MU-MIMO}
\label{sec:asymptotic}
Building on the observation that user crowding severely impairs ZF performance and that STAB alleviates this issue through Doppler separability, we now turn to a formal analysis of these effects. This section derives the average sum rate in an asymptotic regime where $M \to \infty$, providing a clear characterization of the compensatory benefits offered by STAB.

For analytical clarity, we set $\beta_k = 1$ and adopt the small-angle approximations $u_{x, k} \approx \frac{x_k}{2H}$ and $u_{y, k} \approx \frac{y_k}{2H}$, justified by the high satellite altitude. Accordingly, the spatial frequencies are i.i.d. uniform over the support $\CMcal{U} = [-\frac{R}{2H}, \frac{R}{2H}]^2$. Similarly, the normalized Doppler frequency $\omega_k$ for user $k$ is assumed to be i.i.d. uniform over $\CMcal{W} = [-1/2, 1/2]$. The resulting Vandermonde structure in both domains plays a central role in characterizing the scaling behavior of the system.
Under these assumptions, we consider a power-law scaling regime where all relevant system dimensions are defined as fixed powers of the common large parameter $M$:
\begin{itemize}
    \item $M \to \infty$,
    \item $K = M^p$ with $0 < p \le 1$,
    \item $L = M^q$ with $q \ge 0$,
    \item $R/H = M^{-r}$ with $0 < r < 1$.
\end{itemize}

\begin{remark}[Interpretation of the scaling regime] \normalfont
\label{rem:scaling}
The scaling $R/H = M^{-r}$ admits both a practical and an analytical interpretation. 
From a practical standpoint, the exponent $r$ parameterizes the operating point of the LEO satellite communication system for a given antenna count and orbital altitude. A small $r$ corresponds to wide-area coverage where users are spread across many resolution cells (e.g., rural or maritime scenarios), whereas a large $r$ represents a narrow service region in which many users compete for a limited number of spatial degrees of freedom (e.g., dense urban hotspots). For instance, with $M = 256$ antennas at altitude $H = 600$~km, setting $r = 0.5$ yields $R/H \approx 1/16$ and thus $R \approx 37.5$~km, which corresponds to a cell size consistent with practical LEO spot beam footprints.
 
From an analytical perspective, the scaling $R/H = M^{-r}$ ensures that the asymptotic model remains interference-limited. If $R/H$ is fixed as $M \to \infty$, increasing array resolution would eventually separate all users regardless of $K$. 
In that case, the analysis would no longer capture the crowding-induced IUI observed in finite-dimensional systems. 
By shrinking $R/H$ with $M$, the analysis preserves a meaningful balance between user concentration and array resolution, making the asymptotic regime representative of dense finite-dimensional systems. 
\end{remark}

{
\begin{remark}[Interpretation of non-uniform user placement]
\normalfont
Geographical constraints, population density, or localized traffic hotspots may concentrate users within limited angular regions. Such non-uniformity does not change the array resolution, but reduces the number of spatial resolution bins that are effectively occupied. Let $B$ denote the spatial bin count under uniform placement. We introduce a concentration exponent $\mu\geq0$ and define
\begin{align}
B_{\mathrm{eff}}=M^{-\mu}B.
\end{align}
We assume that the effective support contains $B_{\mathrm{eff}}$ bins and that each effective-bin probability is of order $B_{\mathrm{eff}}^{-1}$ up to constant factors. The average bin load then increases from $K/B$ to $K/B_{\mathrm{eff}}$, and the subsequent scaling arguments extend by replacing $B$ with $B_{\mathrm{eff}}$. Thus, non-uniform placement preserves the structure of the scaling laws while adding a crowding exponent that reduces the effectively utilized spatial support.
\end{remark}}

We first present the analysis for a uniform linear array (ULA), where $\mathbf{a}_k = \mathbf{a}(u_{x, k}; M)$ and the support is $\CMcal{U} = [-\frac{R}{2H}, \frac{R}{2H}]$. These results are subsequently extended to the UPA model.

\subsection{Asymptotic Analysis: ULA Case}
\label{subsec:asymptotic_ula}
According to \eqref{eq:average_sum_rate}, the performance of ZF precoding is mainly determined by the channel correlation across users, which reflects their spatial proximity. Specifically, when the channel correlation matrix admits a Vandermonde structure, the performance is determined by the conditioning of the resulting Vandermonde matrix. In particular, since $\mathbf{G}$ is positive definite, letting $\lambda_1(\mathbf{G}) \leq \lambda_2(\mathbf{G}) \leq \cdots \leq \lambda_K(\mathbf{G})$ be the eigenvalues of $\mathbf{G}$, an upper bound on the sum rate can be expressed in terms of the minimum eigenvalue of the Gram matrix as
\begin{align} \label{eq:sumrate_lambda_min_ub}
    R_\Sigma \le M^{p}\log_2 \left(1+\rho M\,\lambda_{1}(\mathbf{G})\right),
\end{align}
where the inequality comes from the fact that $\frac{1}{\operatorname{tr}(\mathbf{G}^{-1})} =\frac{1}{\sum_{i=1}^K \lambda_i(\mathbf{G})^{-1}} \le \lambda_{1}(\mathbf{G})$.

The performance of ZF precoding is fundamentally limited by the spatial resolution of the antenna array: when users are spaced closer than $1/M$ in the spatial domain, the Vandermonde channel matrix becomes severely ill-conditioned. 
In this super-resolution regime, the minimum eigenvalue of the Gram matrix is governed by the size of the most crowded local cluster \cite{batenkov:siam:20, li:stable:2021}. 
Motivated by this observation, we partition the spatial domain into bins of width $1/M$. We define a cluster as the users within a single bin, and its size as the number of those users. 
Under the assumption of i.i.d. uniform user placement, this setup is naturally modeled as a balls-and-bins abstraction, as illustrated in Fig.~\ref{fig:balls_and_bins}.
In particular, this abstraction enables us to characterize the maximum load, denoted by $n_{\max}$, as follows. 

\begin{lemma}[Max load scaling in the balls-and-bins model]\label{lem:max_load}
Assume that $K=M^p$ balls are thrown i.i.d. uniformly at random into $B=M^{1-r}$ bins. As $M\to\infty$, the maximum load $n_{\max}$ scales with high probability (w.h.p.) as
\begin{align}
    n_{\max}=
    \begin{cases}
        O(1), & p<1-r,\\
        \Theta\!\left(\dfrac{\log M}{\log\log M}\right), & p=1-r,\\
        \Theta\!\left(M^{p+r-1}\right), & p>1-r.
    \end{cases}
\end{align}
\end{lemma}
\begin{proof}
    See Appendix~\ref{ap:A}.    
\end{proof}

Subsequently, we reveal the impact of the cluster size on the conditioning of the channel Gram matrix. In particular, we characterize how the minimum eigenvalue deteriorates when $n$ users are packed within the same resolution bin.

\begin{figure}[!t]
    \centerline{\resizebox{0.8\columnwidth}{!}{\includegraphics{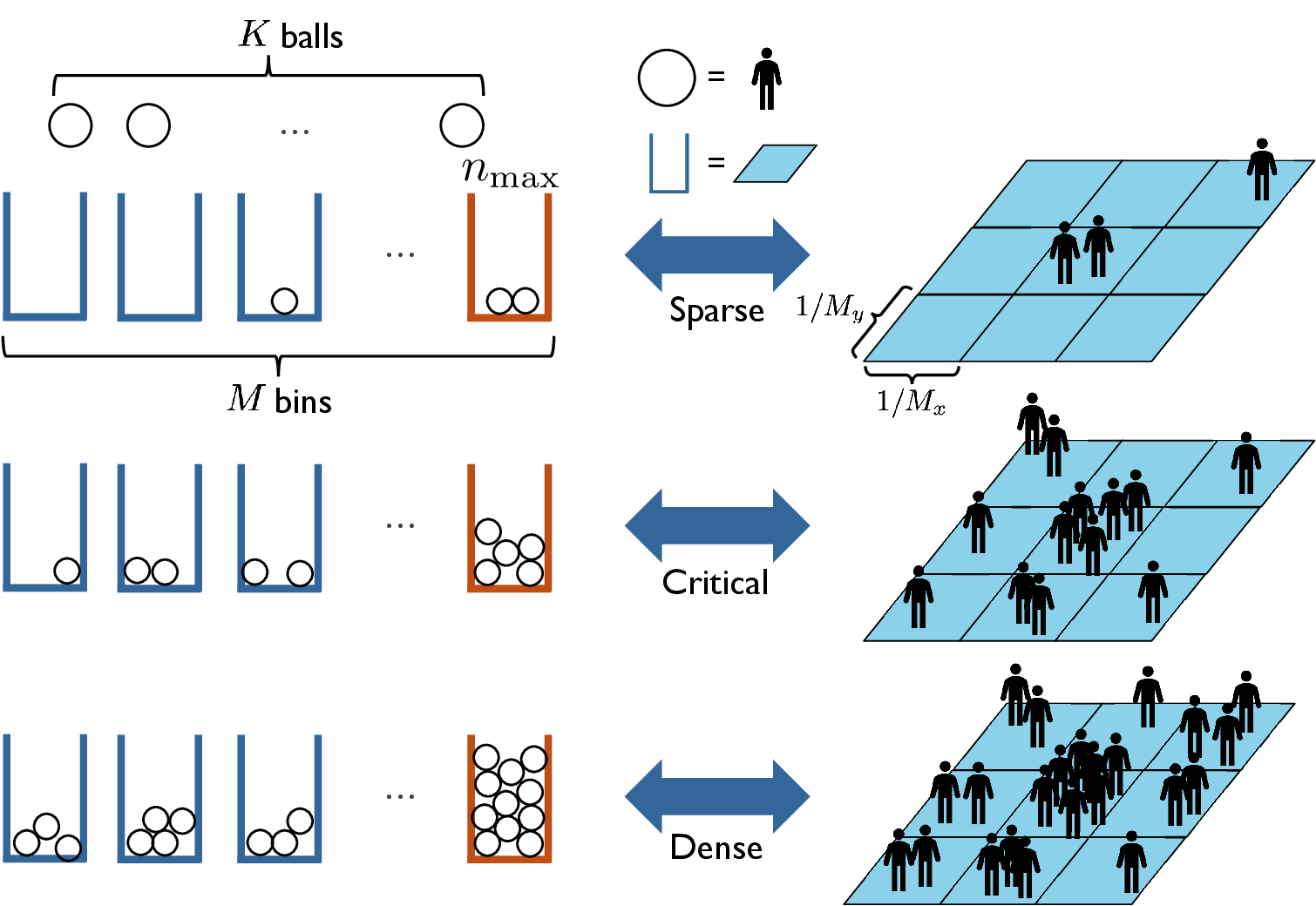}}}
    \caption{Balls-and-bins interpretation of spatial user distribution across sparse, critical, and dense regimes.}
    \label{fig:balls_and_bins}
\end{figure}

\begin{lemma}[Upper bound on $\lambda_1$ under a single cluster]
\label{lem:lambda_min_ub}
{
Suppose that $n\ge2$ users are arbitrarily located within a single resolution bin of dimension $D\in\{1,2,3\}$. Then, their Gram matrix $\mathbf{G}$ satisfies
\begin{align}
\lambda_1(\mathbf{G})
\le
n e^{4\pi D}
\left(
\frac{2\pi D e}{J_D(n)}
\right)^{2J_D(n)},
\label{eq:lambda_min_ub}
\end{align}
where the Taylor cancellation order $J_D(n)$ is defined as
\begin{equation}
J_D(n)
=
\max\left\{
 j\in\mathbb{N}:
 \binom{j-1+D}{D}\le n-1
\right\}.
\label{eq:taylor_order}
\end{equation}
For $D=1$, $J_1(n)=n-1$, and \eqref{eq:lambda_min_ub} reduces to
\begin{align}
\lambda_1(\mathbf{G})
\le
n e^{4\pi}
\left(
\frac{2\pi e}{n-1}
\right)^{2(n-1)}.
\label{eq:lambda_min_ub_1d}
\end{align}
}
\end{lemma}

\begin{proof}
See Appendix~\ref{ap:B}.
\end{proof}
{$D=1$ corresponds to the spatial ULA, $D=2$ to the ULA with STAB or the spatial UPA, and $D=3$ to the UPA with STAB. Lemma~\ref{lem:lambda_min_ub} shows that the minimum eigenvalue decays super-exponentially with the cluster size $n$, uniformly over the user locations within the bin.} This observation is then used in the following lemma to derive an upper bound on the degradation of the ZF sum rate.

\begin{lemma}[ZF sum rate upper bound]\label{lem:zf_sumrate}
Suppose that the maximum load is $n_{\max} = n$. Then, the upper bound of the average sum rate under ZF precoding is given by 
\begin{align}
\mathbb{E}[R_\Sigma  \mid n_{\max} = n] \lesssim M^{p}\log_2 \left(1 + \rho M\,{n\,e^{4\pi}\left(\frac{2\pi e}{n-1}\right)^{2(n-1)}}\right)
\label{eq:zf_sumrate_ub_maxocc}
\end{align}
\end{lemma}
\begin{proof}
See Appendix~\ref{ap:C}.
\end{proof}
Lemma~\ref{lem:zf_sumrate} shows that the single most crowded bin sets the bottleneck. By the interlacing property of Hermitian matrices, adding users from other bins cannot improve the conditioning, so however favorably the remaining users are placed, they cannot undo the cluster that attains $n_{\max}$. {It thus suffices that the small-angle approximation leave the scaling of $n_{\max}$ intact. The approximation enters only through the uniform distribution assumed in Lemma~\ref{lem:max_load}, and under $R/H=M^{-r}$ the true spatial frequency density deviates from uniform by a factor of $1+o(1)$, which does not affect the scaling order.}

Now, we connect the scaling of $n_{\max}$ (Lemma~\ref{lem:max_load}) to the sum rate upper bound of ZF (Lemma~\ref{lem:zf_sumrate}), by which we characterize the ZF performance. Theorem~\ref{thm:ERsum_ULA} is the main result of this section.

\begin{theorem}[Upper scaling law of the average sum rate $R_{\Sigma}$]
\label{thm:ERsum_ULA}
The scaling law of the average sum rate under ZF precoding is upper-bounded as

\textbf{\normalfont{(i) Sparse regime:} $p<1-r$.}
\begin{equation}
\label{eq:ER_sparse}
\mathbb{E}[R_{\Sigma}]\;\lesssim\;M^p\log M.
\end{equation}

\textbf{\normalfont{(ii) Critical regime:} $p=1-r$.}
\begin{equation}
\label{eq:ER_critical}
\mathbb{E}[R_{\Sigma}]\;\lesssim\;
\begin{cases}
M^{r+o(1)}, & r<\frac{1}{2},\\[2pt]
M^{1/2+o(1)}, & r=\frac{1}{2},\\[2pt]
M^{1-r+o(1)}\log M, & r>\frac{1}{2}.
\end{cases}
\end{equation}

\textbf{\normalfont{(iii) Dense regime:} $p>1-r$.}
\begin{equation}
\label{eq:ER_dense}
\mathbb{E}[R_{\Sigma}]\;\to\; 0.
\end{equation}
\end{theorem}
\begin{proof}
See Appendix~\ref{ap:D}.
\end{proof}

Now we interpret Theorem~\ref{thm:ERsum_ULA}.
In the dense regime where $p > 1-r$, the number of users grows faster than the available spatial resolution bins, so that user crowding within each resolution bin becomes unavoidable. 
As a result, the channel Gram matrix becomes increasingly ill-conditioned, and the ZF sum rate collapses to zero as $M \rightarrow \infty$. 
This motivates the use of STAB, which expands the channel into a joint space-Doppler domain and thereby alleviates the crowding-induced ill-conditioning. The following lemma quantifies the eigenvalue behavior under STAB by analyzing a single cluster in the augmented space-Doppler bin.

\begin{lemma}[ZF sum rate upper bound in STAB]\label{lem:zf_sumrate_stab}
Suppose that the maximum load is $n_{\max}=n\ge2$ in the space-Doppler domain. Then, the upper bound of the average sum rate under ZF precoding in STAB is given by
\begin{align}
&\mathbb{E}[\bar{R}_\Sigma \mid n_{\max} = n] \notag\\
&\lesssim M^{p-q}\log_2 \left(1 + \rho M^{1+q}\,{n\,e^{8\pi}\left(\frac{4\pi e}{J_2(n)}\right)^{2J_2(n)}}\right).
\label{eq:zf_sumrate_ub_maxocc_stab}
\end{align}
\end{lemma}
\begin{proof}
The proof follows the same argument as that of Lemma~\ref{lem:zf_sumrate} in Appendix~\ref{ap:C}, {with Lemma~\ref{lem:lambda_min_ub} applied at $D=2$.}
\end{proof}

\begin{theorem}[Upper scaling law of the average sum rate $\bar{R}_{\Sigma}$]
\label{thm:ERsum_STAB}
Consider STAB with $L=M^q$, and let $\delta = p+r-1$ be the scaling threshold. Then, the scaling law of the STAB average sum rate is upper-bounded as

\textbf{\normalfont{(i) Sparse/critical regime:} $q \geq \delta$.}
\begin{equation}
\label{eq:ER_tilde_sparse}
\mathbb{E}[\bar{R}_{\Sigma}]\;\lesssim\;M^{p-q}\log M.
\end{equation}

\textbf{\normalfont{(ii) Dense regime:} $q < \delta$.}
\begin{equation}
\label{eq:ER_tilde_dense}
\mathbb{E}[\bar{R}_{\Sigma}]\;\to\; 0.
\end{equation}
\end{theorem}
\begin{proof}
The proof follows the same argument as that of Theorem~\ref{thm:ERsum_ULA} in Appendix~\ref{ap:D}, after replacing the spatial bin count $B$ with $BL$ and invoking Lemma~\ref{lem:zf_sumrate_stab} in place of Lemma~\ref{lem:zf_sumrate}.
\end{proof}

In addition to the softened eigenvalue decay mentioned in Lemma~\ref{lem:lambda_min_ub}, the gain of STAB admits a clean interpretation in the balls-and-bins framework. Since users are now distinguished by their joint spatial and Doppler signatures, the effective number of bins increases from $B$ to $BL$, which reduces the maximum load $n_{\max}$ for a given $K$. Since $n_{\max}$ governs the minimum eigenvalue of the Gram matrix, this directly translates into improved conditioning of the Gram matrix and a higher achievable sum rate.

In Theorem~\ref{thm:ERsum_STAB}, the parameter $\delta = p + r - 1$ quantifies the excess user load beyond what purely spatial ZF can handle. When $q \ge \delta$, the temporal expansion provided by STAB grows faster than this excess, and the sum rate recovers to a non-vanishing scaling. Conversely, when $q < \delta$, the Doppler dimension is insufficient to compensate for spatial crowding, and the sum rate still collapses.

Note also that $q \le p$ must hold in practice, since the pre-log penalty $1/L$ from symbol repetition eventually dominates if $L$ grows too aggressively. This highlights the separability-efficiency tradeoff inherent in STAB. Since $L=M^{q}$, the tradeoff is governed by the exponent $q$, and the design question reduces to how large $q$ must be to restore the channel conditioning without incurring an unnecessary rate loss.
{
However, this question cannot be settled by the upper bounds obtained so far. Because $R_\Sigma\ge0$, they identify the regimes in which the sum rate collapses, but they do not certify that any choice of $q$ attains the corresponding scaling.
The following corollary establishes achievability.
\begin{corollary}[Achievability in the ULA]
\label{cor:separation}
Let $p>1-r$ and set $q=\delta+\varepsilon$ for any $\varepsilon\in(0,1-r)$. Then,
\begin{equation}
\bar R_\Sigma
=
\Omega\!\left(M^{1-r-\varepsilon}\log M\right)
\quad \text{w.h.p.}
\label{eq:separation_lb}
\end{equation}
Furthermore,
\begin{equation}
\mathbb{E}[\bar R_\Sigma]
=
\Theta\!\left(M^{1-r-\varepsilon}\log M\right).
\label{eq:separation_theta}
\end{equation}
\end{corollary}
\begin{proof}
See Appendix~\ref{ap:E}.
\end{proof}
The condition $0<\varepsilon<1-r$ is equivalent to $\delta<q<p$, and since $p-q=1-r-\varepsilon$, the lower bound \eqref{eq:separation_lb} matches the upper bound in Theorem~\ref{thm:ERsum_STAB}, with the two coinciding as $\varepsilon\to0$. STAB thus achieves the optimal sum rate scaling order in the regime where spatial ZF collapses, and the design question raised above is answered by expanding just beyond the separability threshold $\delta$.
}

\subsection{Asymptotic Analysis: UPA Case}
\label{subsec:asymptotic_upa}
We now extend the analysis to the UPA model. For analytical clarity, we assume $M_x = M_y = M^{1/2}$ so that the UPA provides 2D spatial separability over $(u_{x,k}, u_{y,k})$ for user $k$. To ensure a nontrivial 2D asymptotic regime for the UPA, we focus on the case {$0<r<1/2$}, under which the number of resolution bins grows along both spatial axes. Consequently, the analysis of ZF precoding using only the spatial dimensions in the UPA follows the same 2D Vandermonde structure used for STAB under a ULA, but without the $1/L$ pre-log penalty. Adding STAB to the UPA further introduces a Doppler axis, resulting in a 3D balls-and-bins problem over $(u_{x,k}, u_{y,k}, \omega_k)$.

\begin{figure}[!t]
    \centerline{\resizebox{0.9\columnwidth}{!}{\includegraphics{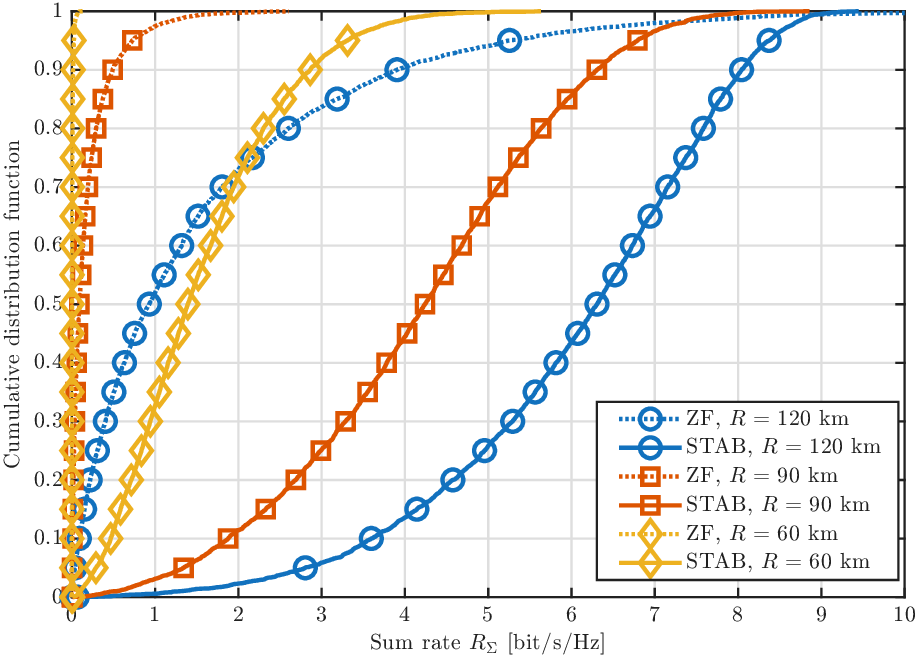}}}
    \caption{Empirical CDFs of the spatial ZF and STAB sum rates under
    the analytical model of Sections~\ref{subsec:system_setup}
    and~\ref{sec:asymptotic}, with $M_x=M_y=16$, $K=16$, $L=3$,
    $H=600$ km, $P=28$ dBm, and different $R$, where
    $|\beta_k|^2$ follows the Friis model~\protect\cite{friis:note:1964}
    with $f_c=2$ GHz.}
    \label{fig:cdf}
\end{figure}

Specifically, under the small-angle approximation, the spatial frequencies $(u_{x,k}, u_{y,k})$ are i.i.d. uniform over the 2D support $\CMcal{U} = \CMcal{U}_x \times \CMcal{U}_y$, where $\CMcal{U}_x = \CMcal{U}_y = \left[-\frac{R}{2H}, \frac{R}{2H}\right]$.
Partitioning $\CMcal{U}$ into bins of size $(1/M_x) \times (1/M_y)$ yields a total of
\begin{align}
B = \frac{|\CMcal{U}_x|}{1/M_x} \cdot \frac{|\CMcal{U}_y|}{1/M_y} = M \left(\frac{R}{H}\right)^2 = M^{1-2r}
\label{eq:upa_bin_count}
\end{align}
spatial bins, so the system reduces to a balls-and-bins problem with $K = M^p$ balls and $B = M^{1-2r}$ bins.  
Since the 2D single cluster analysis is identical to that of STAB with $(u_{x,k},\omega_k)$ replaced by $(u_{x,k},u_{y,k})$, the ZF sum rate upper bound for the UPA is
\begin{align}
R_\Sigma \lesssim M^p \log_2 \left( 1 + \rho M \,{n\,e^{8\pi}\left(\tfrac{4\pi e}{J_2(n)}\right)^{2J_2(n)}} \right),
\label{eq:zf_sumrate_upa}
\end{align}
where the pre-log term has no $1/L$ penalty since no temporal expansion is used.

Based on \eqref{eq:zf_sumrate_upa}, we establish the asymptotic upper scaling law for the average ZF sum rate under the UPA as follows.
\begin{theorem}[Upper scaling law of the average sum rate $R_{\Sigma}$ for UPA]
\label{thm:ERsum_UPA}
Consider the UPA baseline without STAB.

\textbf{\normalfont{(i) Sparse/critical regime: $p + 2r - 1 \le 0$.}}
\begin{align}
\mathbb{E}[R_{\Sigma}] \;\lesssim\; M^{p} \log M.
\end{align}

\textbf{\normalfont{(ii) Dense regime: $p + 2r - 1 > 0$.}}
\begin{align}
\mathbb{E}[R_{\Sigma}] \;\to\; 0.
\end{align}
\end{theorem}

\begin{proof}
The proof follows the balls-and-bins and cluster-conditioning arguments used in Theorem~\ref{thm:ERsum_STAB}, substituting the space-Doppler bin count by the 2D spatial bin count $B = M^{1-2r}$ and omitting the $1/L$ pre-log penalty.
\end{proof}

For STAB under UPA, the Doppler axis extends the 2D spatial domain to a 3D space-Doppler domain $(u_{x,k}, u_{y,k}, \omega_k)$ over the support $\CMcal{U}_x \times \CMcal{U}_y \times \CMcal{W}$. Partitioning the joint support into bins of size $(1/M_x) \times (1/M_y) \times (1/L)$ yields $M^{1-2r+q}$ space-Doppler bins. Following the same argument as in the previous analysis, an upper bound on the average sum rate of STAB in UPA is given by
\begin{align} \label{eq:zf_sumrate_upa_stab}
\bar R_\Sigma \lesssim M^{p-q} \log_2\left( 1+\rho M^{1+q} \,{n\,e^{12\pi}\left(\tfrac{6\pi e}{J_3(n)}\right)^{2J_3(n)}} \right),
\end{align}
where $M^{p-q}$ captures the $1/L$ pre-log penalty incurred by using $L=M^q$ time resources. Based on \eqref{eq:zf_sumrate_upa_stab}, the asymptotic upper bound on the STAB average sum rate for the UPA setting scales as follows.

\begin{theorem}[Upper scaling law of the average sum rate $\bar{R}_{\Sigma}$ for UPA]
\label{thm:ERsum_UPA_STAB}
Consider STAB with $L=M^q$, and let $\delta_{\mathrm{UPA}} = p+2r-1$ be the scaling threshold in UPA. Then, the scaling law of the STAB average sum rate is upper-bounded as

\textbf{\normalfont{(i) Sparse/critical regime:} $q \ge \delta_{\mathrm{UPA}}$.}
\begin{align}
\mathbb E[\bar R_{\Sigma}]
\;\lesssim\;
M^{p-q}\log M.
\label{eq:ER_UPA_STAB_sparsecritical}
\end{align}

\textbf{\normalfont{(ii) Dense regime:} $q < \delta_{\mathrm{UPA}}$.}
\begin{align}
\mathbb E[\bar R_{\Sigma}]
\;\to\; 0.
\label{eq:ER_UPA_STAB_dense}
\end{align}
\end{theorem}

\begin{proof}
The proof follows the same balls-and-bins conditioning argument as in the previous STAB analysis, after replacing the spatial bin count $M^{1-2r}$ with the space-Doppler bin count $M^{1-2r+q}$ and applying the 3D single cluster eigenvalue bound underlying \eqref{eq:zf_sumrate_upa_stab}.
\end{proof}

\begin{figure}[!t]
    \centerline{\resizebox{0.9\columnwidth}{!}{\includegraphics{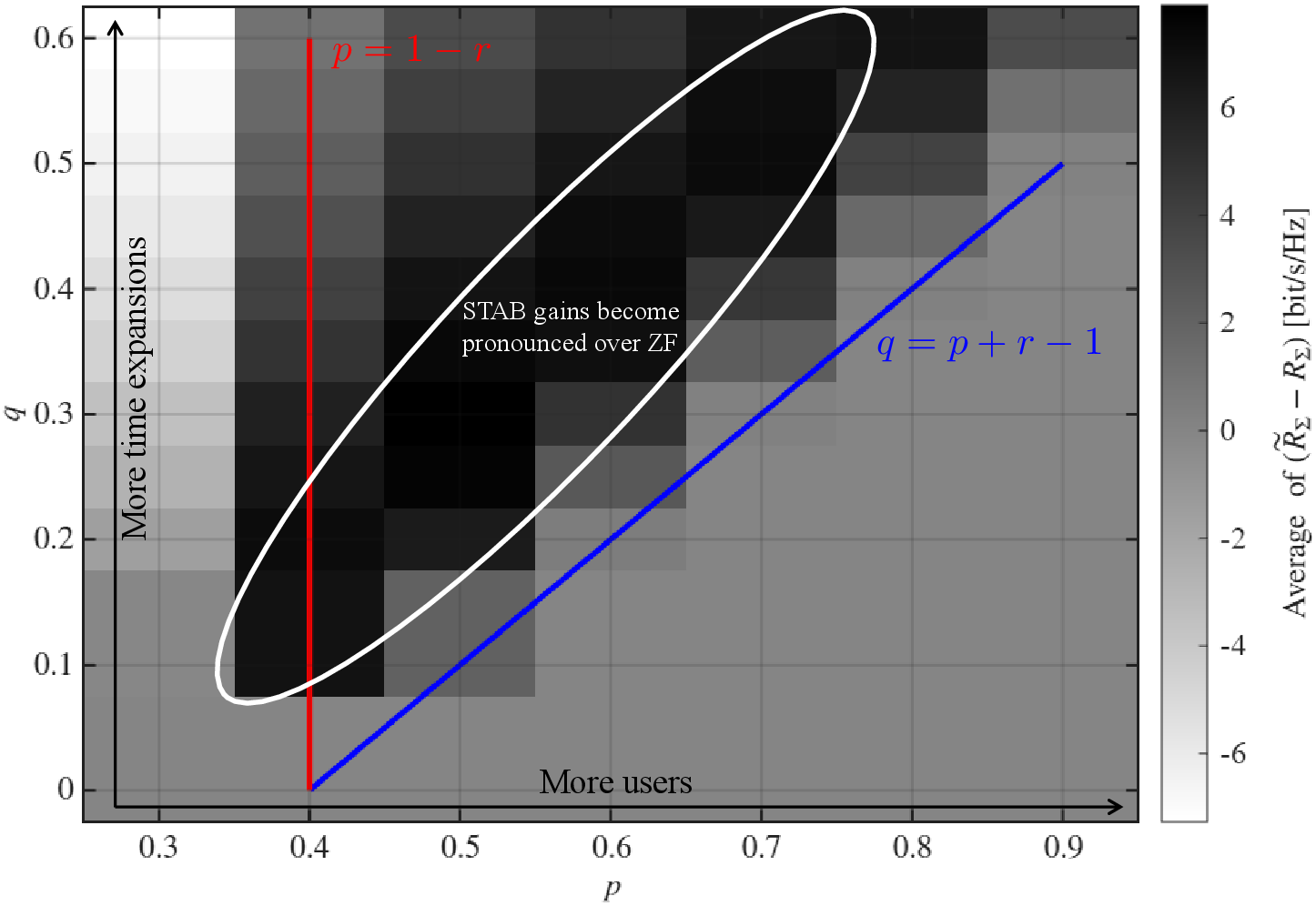}}}
    \caption{Heat map of $\mathbb{E}[\bar{R}_{\Sigma}-R_{\Sigma}]$ for ULA over the $(p,q)$-plane with $M=256$, $r=0.60$, $K=M^p$, and $L=M^q$, under the same parameters as in Fig.~\ref{fig:cdf}.}
    \label{fig:gain_map}
\end{figure}

The structure of Theorem~\ref{thm:ERsum_UPA_STAB} mirrors that of Theorem~\ref{thm:ERsum_STAB}, with the excess load parameter now given by $\delta_\mathrm{UPA} = p + 2r - 1$ instead of $\delta = p + r - 1$. The factor $2r$ in the UPA case arises because $r$ describes the scaling of the linear spatial extent $R/H \sim M^{-r}$. For the 2D spatial support, the relevant measure becomes $(R/H)^2 \sim M^{-2r}$, leading naturally to the modified crowding law and the threshold $p+2r-1$.{
Theorem~\ref{thm:ERsum_UPA_STAB} is again an upper bound, and the corresponding achievability result is as follows.
\begin{corollary}[Achievability in the UPA]
\label{cor:separation_upa}
Let $r < \frac{1}{2}$, $p > 1-2r$, and set $q = \delta_{\mathrm{UPA}} + \varepsilon$ for any $\varepsilon \in (0, 1-2r)$. Then,
\begin{equation}
\mathbb{E}[\bar R_\Sigma] = \Theta\!\left(M^{1-2r-\varepsilon}\log M\right).
\label{eq:separation_theta_upa}
\end{equation}
\end{corollary}

\begin{proof}
See Appendix~\ref{ap:E}.
\end{proof}
}

Figs.~\ref{fig:cdf} and~\ref{fig:gain_map} provide finite-dimensional validation of the scaling laws. In Fig.~\ref{fig:cdf}, $p=0.5$ and $q\approx0.20$, while $R=60$, $90$, and $120$ km correspond to $r\approx0.415$, $0.342$, and $0.290$, respectively. Thus, increasing $R$ decreases $r$ toward the UPA threshold $(1-p)/2=0.25$ of spatial ZF, whose sum rate begins to recover at $R=120$ km. In contrast, the larger STAB threshold $(1+q-p)/2\approx0.35$ allows STAB to avoid collapse at $R=90$ and $120$ km, while its recovery remains limited at $R=60$ km. Consistently, Fig.~\ref{fig:gain_map} shows that the STAB gain is concentrated around the region $p>1-r$ and $q\ge p+r-1$, where the spatial domain is dense but the joint space-Doppler domain is sparse or critical. These results show that the asymptotic regime boundaries remain visible at moderate values of $M$, $K$, and $L$.

\section{Space-Doppler User Selection for STAB}
\label{sec:joint_stab_scheduling}
Due to the crowding-induced eigenvalue collapse identified in our analysis, the spatial multiplexing gains in LoS-dominant LEO downlinks are fundamentally limited. This makes effective user scheduling essential, yet conventional semi-orthogonal user selection (SUS) \cite{yoo:jsac:06} operates solely in the spatial domain and cannot exploit the Doppler separability introduced by STAB. 
Motivated by this, in this section, we propose the SDS algorithm, which extends the semi-orthogonality criterion from the spatial domain to the joint space-Doppler domain. By jointly resolving users who are spatially proximate but possess distinct Doppler characteristics, SDS achieves higher multiplexing efficiency compared to purely spatial scheduling. 

SDS performs greedy sequential user selection based on the joint space-Doppler channel vectors $\bar{\mathbf{h}}_k$. Starting from an empty scheduled set $\CMcal{S}$ and candidate pool $\CMcal{K}_1 = \{1, \dots, U\}$, the algorithm computes at each iteration $i$ the effective channel vector
\begin{align}
\mathbf{g}_k = (\mathbf{I} - \mathbf{Q}_{i-1}\mathbf{Q}_{i-1}^{\mathsf{H}})\bar{\mathbf{h}}_k, \quad k \in \CMcal{K}_i
\end{align}
where $\mathbf{Q}_{i-1} = [\mathbf{q}_1, \dots, \mathbf{q}_{i-1}]$ collects the orthonormal bases from previously selected users, with $\mathbf{Q}_0 = \emptyset$ for the first iteration.

To maximize the available space-Doppler degrees of freedom, SDS selects the user $\pi(i)$ who yields the largest effective channel gain $\|\mathbf{g}_k\|^2$ from the current candidate set $\CMcal{K}_i$. After adding $\pi(i)$ to $\CMcal{S}$, the orthonormal basis matrix is updated to $\mathbf{Q}_i = [\mathbf{Q}_{i-1}, \mathbf{q}_i]$ with $\mathbf{q}_i = \mathbf{g}_{\pi(i)}/\|\mathbf{g}_{\pi(i)}\|$. 
The candidate set is then filtered by the space-Doppler semi-orthogonality criterion
\begin{align}
\frac{|\bar{\mathbf{h}}_k^{\mathsf{H}} \bar{\mathbf{h}}_{\pi(i)}|}{\|\bar{\mathbf{h}}_k\| \|\bar{\mathbf{h}}_{\pi(i)}\|} < \alpha_{\mathrm{ST}} \label{eq:criterion}
\end{align}
where $\alpha_{\mathrm{ST}}$ is a numerically optimized threshold.

The complete procedure of the proposed SDS is summarized in Algorithm~\ref{alg:stsus}. SDS constructs a well-conditioned space-time channel matrix without exhaustive search, and terminates when the number of selected users reaches $K$ or no remaining candidate satisfies the semi-orthogonality condition.
{
Selecting $\pi(i)$ favors the channel least aligned with the span of the already selected users, suppressing the growth of $\operatorname{tr}(\bar{\mathbf{G}}^{-1})$, while \eqref{eq:criterion} encourages the selected users to occupy distinct bins, lowering $n_{\max}$. In LEO downlinks where spatial angular separation is inherently limited, SDS therefore not only supports more simultaneous users than purely spatial scheduling, but also prevents the rate collapse induced by a large $n_{\max}$.

In terms of complexity, SDS costs $O(UMLK^2)$, exceeding the $O(UMK^2)$ of purely spatial SUS by the factor $L$. This factor arises from the $ML$-dimensional inner products of $\bar{\mathbf{h}}_k$, and is removed by replacing them with the product of the Dirichlet kernels in closed form, which brings SDS to the same order of complexity as SUS. Moreover, since every input of SDS is determined by the geometric parameters, the selection is recomputed only once per coherence interval of these parameters, which far exceeds the symbol duration and thus renders SDS feasible on LEO satellite payloads.
}

\begin{algorithm}[t]
\caption{SDS Algorithm}
\label{alg:stsus}
\BlankLine
\textbf{Initialize}: $i \gets 1$, $\CMcal{S} \gets \emptyset$, $\CMcal{K}_1 \gets \{1,2,\ldots,U\}$, and $\mathbf{Q} \gets [\,]$\;
\While{$i \le K$ \textbf{and} $\CMcal{K}_i \neq \emptyset$}{
    \ForEach{$k \in \CMcal{K}_i$}{
        \eIf{$i = 1$}{
            $\mathbf{g}_k \gets \bar{\mathbf{h}}_k$\;
        }{
            $\mathbf{g}_k \gets (\mathbf{I} - \mathbf{Q}\mathbf{Q}^{\mathsf H})\bar{\mathbf{h}}_k$\;
        }
    }
    $\pi(i) \gets \arg\max_{k \in \CMcal{K}_i} \|\mathbf{g}_k\|^2$\;
    $\CMcal{S} \gets \CMcal{S} \cup \{\pi(i)\}$\;
    $\mathbf{q}_i \gets \mathbf{g}_{\pi(i)} / \|\mathbf{g}_{\pi(i)}\|$\;
    $\mathbf{Q} \gets [\mathbf{Q}, \mathbf{q}_i]$\;
    \If{$i < K$}{
        $\CMcal{K}_{i+1} \gets \left\{ k \in \CMcal{K}_i \setminus \{\pi(i)\} :
        \dfrac{|\bar{\mathbf{h}}_k^{\mathsf H}\bar{\mathbf{h}}_{\pi(i)}|}
        {\|\bar{\mathbf{h}}_k\|\,\|\bar{\mathbf{h}}_{\pi(i)}\|} < \alpha_{\mathrm{ST}} \right\}$\;
    }
    $i \gets i + 1$\;
}
\end{algorithm}

{
\section{Robustness of the Scaling Laws}
\label{sec:robust}
In this section, we examine whether STAB preserves its scaling gains under CSI errors and residual multipath. Since STAB reconstructs each LoS space-time signature from the spatial frequencies and normalized Doppler, errors from estimation, quantization, and feedback delay manifest as perturbations of these parameters.

We aggregate these effects as $\widehat{u}_{x,k}=u_{x,k}+\Delta u_{x,k}$, $\widehat{u}_{y,k}=u_{y,k}+\Delta u_{y,k}$, and $\widehat{\omega}_k=\omega_k+\Delta\omega_k$, where the aggregate errors are characterized by $\mathbb{E}[(\Delta u_{x,k})^2]=\mathbb{E}[(\Delta u_{y,k})^2]=\sigma_u^2$ and $\mathbb{E}[(\Delta\omega_k)^2]=\sigma_\omega^2$. Let
$\bar{\mathbf a}_k=\mathbf b_k\otimes\mathbf a_k$
denote the unit-norm LoS space-time signature, and let
$\widehat{\bar{\mathbf a}}_k$ denote its reported counterpart.

\subsection{Imperfect CSI}
\label{sec:rob_csi}

The fraction of the true LoS energy lying outside the reported direction is
\begin{equation}
\zeta_{k,\mathrm{geo}}
=
1-\left|
\widehat{\bar{\mathbf a}}_k^{\mathsf H}
\bar{\mathbf a}_k
\right|^2.
\label{eq:rob_geo}
\end{equation}
In the small-error regime, expanding each Dirichlet kernel to second order,
\begin{equation}
\mathbb{E}[\zeta_{k,\mathrm{geo}}]
\approx
\frac{\pi^2}{3}
\left[
(M_x^2+M_y^2-2)\sigma_u^2
+
(L^2-1)\sigma_\omega^2
\right].
\label{eq:rob_geo_avg}
\end{equation}
Thus, CSI robustness is governed by the resolution-normalized errors
$M_x\sigma_u$, $M_y\sigma_u$, and $L\sigma_\omega$, rather than by the raw parameter errors.
We parameterize
\begin{equation}
\sigma_u=M^{-s_u},
\qquad
\sigma_\omega=M^{-s_\omega}.
\label{eq:rob_csi_parameterization}
\end{equation}
For the square-UPA scaling
$M_x=M_y=M^{1/2}$ and $L=M^q$,
\eqref{eq:rob_geo_avg} gives
\begin{equation}
\mathbb{E}[\zeta_{k,\mathrm{geo}}]
\lesssim
M^{1-2s_u}
+
M^{2q-2s_\omega}.
\label{eq:rob_geo_scaling}
\end{equation}

\subsection{Residual Multipath and Unified Robustness}
\label{sec:rob_mp}

Residual multipath is modeled through the composite channel
\begin{equation}
\bar{\mathbf h}_{k,\kappa}
=
\sqrt{ML}\beta_k
\left(
\sqrt{\tfrac{\kappa_k}{1+\kappa_k}}\,\bar{\mathbf a}_k
+
\sqrt{\tfrac{1}{1+\kappa_k}}\,\boldsymbol{\xi}_k
\right),
\label{eq:rob_composite}
\end{equation}
where the diffuse components $\boldsymbol{\xi}_k$ are independent
across users and of the LoS components, with
$\mathbb{E}[\boldsymbol{\xi}_k]=\mathbf 0$,
$\mathbb{E}[\boldsymbol{\xi}_k\boldsymbol{\xi}_k^{\mathsf H}]
=\boldsymbol{\Sigma}_{\xi,k}$, and
$\operatorname{tr}(\boldsymbol{\Sigma}_{\xi,k})=1$.
With
$\mathbf P_k^\perp=\mathbf I-\widehat{\bar{\mathbf a}}_k
\widehat{\bar{\mathbf a}}_k^{\mathsf H}$
denoting the projector onto the orthogonal complement of the reported
LoS direction, the diffuse leakage is
$\zeta_{k,\mathrm{mp}}
=\mathbb{E}[\|\mathbf P_k^\perp\boldsymbol{\xi}_k\|^2]
=\operatorname{tr}(\mathbf P_k^\perp\boldsymbol{\Sigma}_{\xi,k})$.

Since the scatterers are localized near each user, the paths of user
$k$ approximately share the satellite-side steering vector
$\mathbf a_k$, and we adopt
\begin{equation}
\boldsymbol{\xi}_k
=
\mathbf b_{k,\mathrm d}\otimes\mathbf a_k,
\label{eq:rob_diffuse_model}
\end{equation}
with
$\mathbf b_{k,\mathrm d}\sim
\mathcal{CN}(\mathbf 0,\mathbf I_L/L)$.
Under this decorrelated slow-time model, $\boldsymbol{\xi}_k$ has unit
average energy while
$\mathbb{E}[|\widehat{\bar{\mathbf a}}_k^{\mathsf H}
\boldsymbol{\xi}_k|^2]
=|\widehat{\mathbf a}_k^{\mathsf H}\mathbf a_k|^2/L$,
so that
\begin{equation}
\zeta_{k,\mathrm{mp}}
=
1-\frac{1}{L}
\left|\widehat{\mathbf a}_k^{\mathsf H}\mathbf a_k\right|^2
\approx
1-\frac{1}{L},
\label{eq:rob_mp_fraction}
\end{equation}
the additional leakage from spatial mismatch being already captured by
$\zeta_{k,\mathrm{geo}}$.
Projecting \eqref{eq:rob_composite} onto $\mathbf P_k^\perp$ and noting
that the LoS--diffuse cross term vanishes in expectation, the total
normalized off-direction channel energy is
\begin{equation}
\zeta_{k,\mathrm{tot}}
=
\frac{\kappa_k\zeta_{k,\mathrm{geo}}+\zeta_{k,\mathrm{mp}}}
{1+\kappa_k}.
\label{eq:rob_total}
\end{equation}

Let $\widehat{\bar{\mathbf G}}$ denote the Gram matrix of the reported
channels. ZF eliminates interference along the reported directions,
while the residual off-direction interference power is bounded by
$\zeta_{k,\mathrm{tot}}$ times the squared ZF precoder norm, which
scales as $1/\lambda_1(\widehat{\bar{\mathbf G}})$. We therefore define
\begin{equation}
\Gamma_k
=
\frac{\zeta_{k,\mathrm{tot}}}
{\lambda_1(\widehat{\bar{\mathbf G}})}.
\label{eq:rob_metric}
\end{equation}
Provided that the desired channel projection onto the reported beam
remains bounded away from zero,
\begin{equation}
\frac{1}{\mathrm{SINR}_{k,\mathrm{imp}}}
\lesssim
\frac{1}{\overline{\mathrm{SINR}}_k}
+
\Gamma_k,
\label{eq:rob_sinr}
\end{equation}
i.e., the same CSI or multipath leakage is relatively benign for
well-conditioned reported channels but strongly amplified when
$\widehat{\bar{\mathbf G}}$ is nearly singular.

Parameterizing the effective multipath leakage and reported-channel
conditioning as
$(1-1/L)/(1+\kappa_k)=M^{-s_{\mathrm{mp}}}$ and
$\lambda_1(\widehat{\bar{\mathbf G}})=M^{-s_\lambda}$, and using the
UPA bound
$\mathbb{E}[\zeta_{k,\mathrm{geo}}]
\lesssim M^{1-2s_u}+M^{2q-2s_\omega}$
from \eqref{eq:rob_geo_scaling}, relations
\eqref{eq:rob_total} and \eqref{eq:rob_metric} yield
\begin{equation}
\mathbb{E}[\Gamma_k]
\lesssim
M^{s_\lambda+1-2s_u}
+
M^{s_\lambda+2q-2s_\omega}
+
M^{s_\lambda-s_{\mathrm{mp}}}.
\label{eq:rob_scaling}
\end{equation}
Hence $\Gamma_k$ decays polynomially whenever
$s_u>(1+s_\lambda)/2$,
$s_\omega>q+s_\lambda/2$, and
$s_{\mathrm{mp}}>s_\lambda$.
Under these conditions, each scheduled user retains a rate of order
$\log M$,
\begin{equation}
\mathbb{E}[\bar R_{\Sigma,\mathrm{imp}}]
=
\Omega\!\left(M^{p-q}\log M\right),
\label{eq:rob_rate}
\end{equation}
showing that imperfect CSI and residual multipath do not change the
ideal STAB sum rate order.
}

\section{Numerical Results}
\label{sec:numerical}
\subsection{Simulation Setup}
\label{subsec:simulation_setup}
{
We consider a Ka-band LEO downlink serving earth stations in motion (ESIM)~\cite{itu:wrc23:res123}. The satellite operates at an altitude of $H=600$~km and employs a UPA with $M_x=M_y=16$. Following Section~\ref{subsec:system_setup}, users are distributed over the service area with a cell size of $R=75$~km, corresponding to $R/H=0.125$. By \eqref{eq:upa_bin_count}, this geometry yields $B = M(R/H)^2 = 4$ spatial resolution bins. In each scheduling instance, the satellite selects and simultaneously serves $K=32$ users from $U=256$ candidates.
The cell boresight elevation angle is $\theta_{\mathrm{el}}=30^\circ$, and the carrier frequency and bandwidth are $f_c=20$~GHz and $6$~MHz, respectively. The receiver noise power is computed from a noise spectral density of $-174$~dBm/Hz. Each terminal has a receive gain of $G_{\mathrm{rx}}=36.8$~dBi~\cite{itur_s2464,itur_res169}. Under the Friis free-space model, the large-scale channel amplitude is
\begin{equation} |\beta_k| = \sqrt{G_{\mathrm{rx}} \left(\frac{c}{4\pi f_c d_k} \right)^2}, \label{eq:simulation_path_gain} \end{equation}
where the path loss exponent is two, $c$ is the speed of light, and $d_k$ is the slant range to user $k$.
For a terminal moving at speed $v_k$ with heading $\psi_k$ relative to the satellite azimuth, the residual Doppler shift is
\begin{equation}
\nu_k^{\mathrm{ut}}
=
\frac{f_c}{c}
v_k\cos\theta_{\mathrm{el}}\cos\psi_k.
\label{eq:esim_doppler}
\end{equation}
The headings and speeds are independently drawn as
$\psi_k\sim\mathcal{U}[0,2\pi)$ and
$v_k\sim\mathcal{U}[0,v_{\max}]$, respectively, where
$v_{\max}=250$~m/s. This bounded uniform speed model is commonly adopted as a stochastic mobility abstraction in satellite-network studies~\cite{abdollahpour:spawc:26,sharma:twc:20}.
With
$\nu_{\max}=(f_c/c)v_{\max}\cos\theta_{\mathrm{el}}$,
the snapshot interval is set to
$T_r=1.5/(2\nu_{\max})=\SI{52}{\micro\second}$.
Unless otherwise stated, STAB uses $L=3$ snapshots, yielding an observation interval of
$LT_r=\SI{156}{\micro\second}$.
The large-scale coefficient $\beta_k$ and spatial steering vector $\mathbf a_k$ are treated as constant over this interval.

\begin{figure}[!t]
    \centerline{\resizebox{0.9\columnwidth}{!}{
    \includegraphics{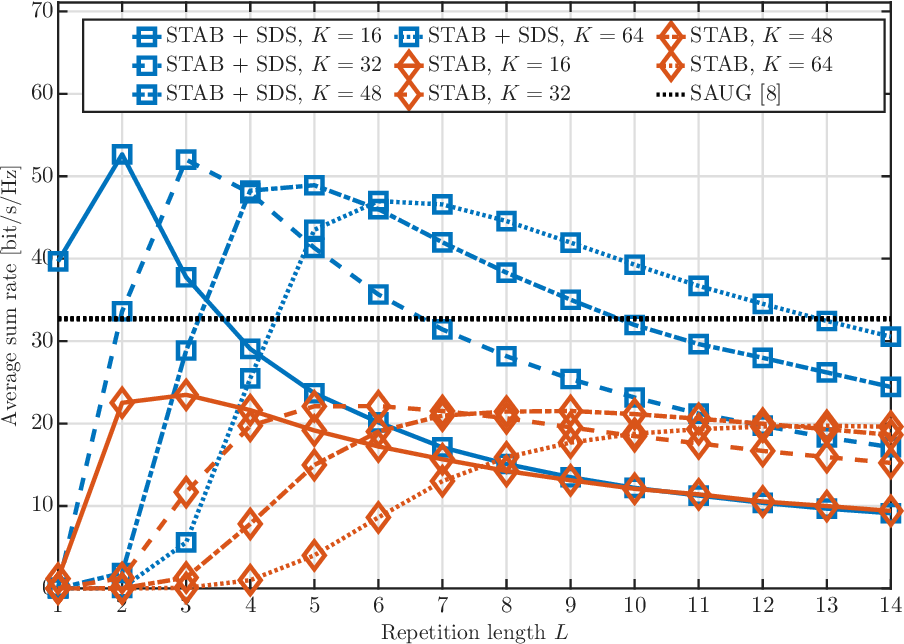}}}
    \caption{Average sum rate versus the repetition length $L$ for different
    served loads with $\kappa_k=20$~dB, $P=44$~dBm.}
    \label{fig:optimal_L}
\end{figure}

\begin{figure}[!t]
    \centerline{\resizebox{0.9\columnwidth}{!}{
    \includegraphics{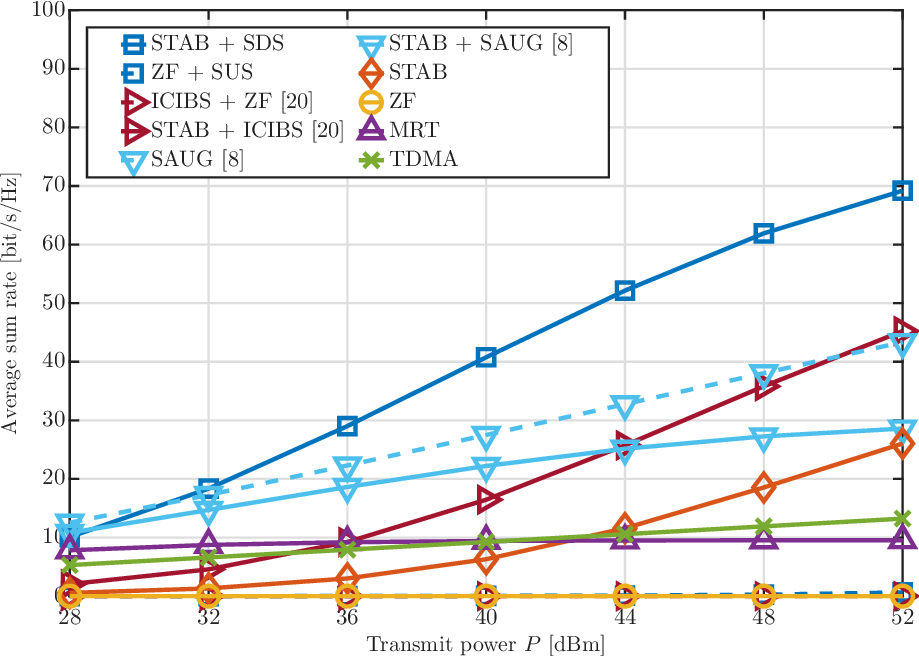}}}
    \caption{Average sum rate versus transmit power for STAB + SDS, the
    spatial baselines, and their space-Doppler extensions, with $L=3$ and
    $\kappa_k=20$ dB.}
    \label{fig:sumrate_power}
\end{figure}

\begin{remark}[Physical Doppler and the analytical model]
\normalfont
Taking the normalized Doppler $\omega_k$ modulo 1 folds the physical Doppler axis onto $\CMcal{W}$, and each point in $\CMcal{W}$ accumulates the density of all physical Doppler values mapped to it. As $T_r$ increases, more and more widely separated segments of the physical density overlap at each point, and this averaging flattens the density into a uniform distribution over $\CMcal{W}$. In other words, while flattening requires a long $T_r$, maintaining the stationarity of the channel coefficients requires a short $T_r$. A high carrier frequency, such as the Ka-band, resolves this tradeoff by proportionally widening the Doppler spread, allowing wrapping to occur even with a short $T_r$. Accordingly, the uniform normalized Doppler model in Section~\ref{sec:asymptotic} is used as a limiting analytical approximation, whereas the numerical simulations generate Doppler shifts directly from the geometric mobility model. Furthermore, since the carrier frequency is involved only in this mapping to generate $\omega_k$, the results in Section~\ref{sec:asymptotic} are independent of it; the carrier frequency solely determines the velocity difference to which a normalized Doppler separation corresponds.
\end{remark}
The non-ideal channels follow Section~\ref{sec:robust}. We set $M_x\sigma_u = M_y\sigma_u = L\sigma_\omega \in \{10^{-3},10^{-2}\}$, where the underlying errors $\Delta u_{x,k}$, $\Delta u_{y,k}$, and $\Delta\omega_k$ are each generated as i.i.d.\ Gaussian random variables, and the residual multipath follows the decorrelated diffuse model in \eqref{eq:rob_diffuse_model} with $\kappa_k=20$~dB.
}

\begin{figure}[!t]
    \centerline{\resizebox{0.9\columnwidth}{!}{
    \includegraphics{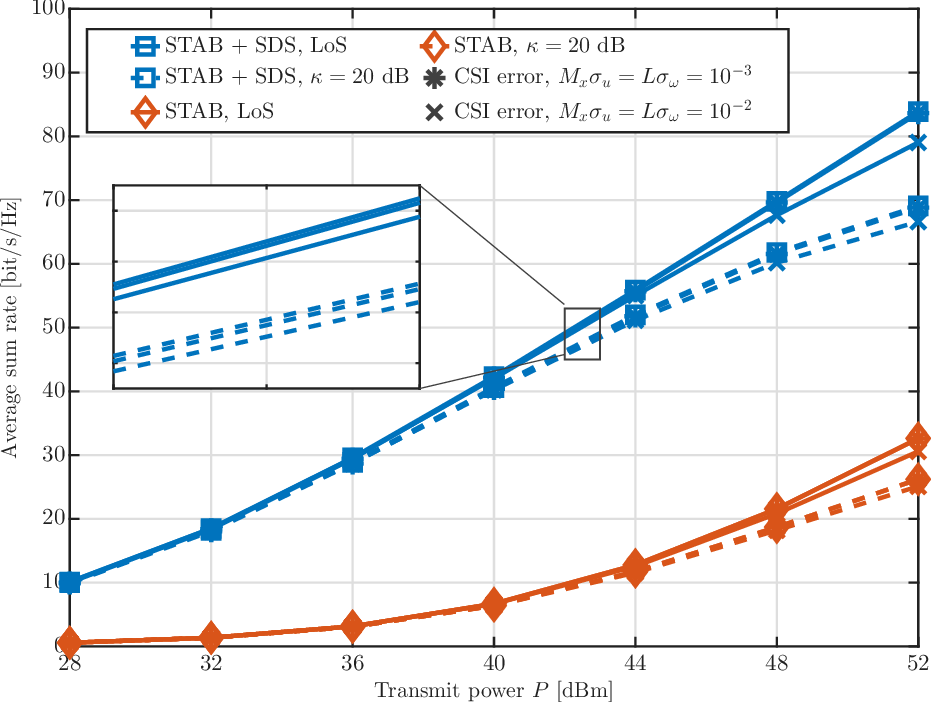}}}
    \caption{Average sum rate versus transmit power under the residual multipath and CSI errors.}
    \label{fig:robustness}
\end{figure}

\subsection{Performance Evaluations}
\label{subsec:finite_dimensional_results}
{
\fig{fig:optimal_L} illustrates the tradeoff between additional space-Doppler resolution and the $1/L$ repetition penalty. Increasing $L$ initially improves channel conditioning by spreading users over more joint resolution bins, but the repetition penalty eventually dominates. The optimal repetition lengths are approximately $L^\star=2, 3, 5$, and $6$ for $K=16, 32, 48$, and $64$, respectively, maintaining the average joint bin occupancy $K/(BL^\star)$ within a narrow range of $2$--$2.7$. Thus, $L^\star$ scales with the spatial crowding level $K/B$ and should be chosen as the shortest length that provides sufficient conditioning recovery. Furthermore, since space-angle user grouping (SAUG)~\cite{you:jsac:20} always divides the pre-log term by $n_{\max}$, STAB consistently achieves superior performance at the optimal $L$.

\fig{fig:sumrate_power} compares STAB + SDS with spatial baselines and their space-Doppler extensions. Among schemes without user selection, STAB already outperforms ZF, maximum ratio transmission (MRT), and time division multiple access (TDMA), as the Doppler dimension alleviates the severe spatial correlation caused by the narrow angular spread at the satellite. Among selection-based baselines, inter-channel interference-based selection (ICIBS)~\cite{chaves:ojcoms:22} and ZF + SUS rely mainly on pairwise channel correlations, failing to reliably prevent Gram matrix ill-conditioning under limited spatial degrees of freedom.

SAUG avoids this collapse by assigning users in the same spatial bin to orthogonal groups, but incurs a pre-log penalty proportional to $1/n_{\max}$. In contrast, STAB distributes users over joint space-Doppler bins, enabling a repetition length smaller than the required number of spatial groups while multiplexing more streams per channel use. SDS further enhances the minimum Gram matrix eigenvalue by selecting a well-conditioned user set. Consequently, STAB + SDS achieves the highest sum rate above $32\text{ dBm}$, with gains becoming more pronounced under higher occupancy where channel conditioning dominates performance.

Applying STAB to the baselines further clarifies these behaviors. STAB + ICIBS benefits from Doppler resolution, but its normalized pairwise correlation metric neither accounts for weak large-scale gains nor directly controls the worst-conditioned direction of the Gram matrix. Conversely, STAB offers negligible gain to SAUG because its groups are already separated via orthogonal resource allocation while the $1/L$ penalty remains.

\fig{fig:robustness} evaluates the effects of CSI errors and residual multipath. Larger $M_x\sigma_u$, $M_y\sigma_u$, and $L\sigma_\omega$ increase the mismatch between the true and reported LoS signatures, while a smaller Rician factor increases the diffuse power outside the reported direction. Both effects increase $\Gamma_k$ and hence the residual ZF leakage, particularly at high transmit power. The loss is larger under the decorrelated diffuse model because approximately $1-1/L$ of the diffuse power lies outside the nominal Doppler signature. Nevertheless, STAB + SDS remains more robust than unselected STAB because SDS improves $\lambda_1(\widehat{\bar{\mathbf G}})$ and limits the amplification of the same channel mismatch.
}

\section{Conclusion}
\label{sec:conclusion}
This paper established the fundamental performance limits of MU-MIMO in LoS-dominant LEO satellite channels by linking user crowding to the conditioning of the Vandermonde channel matrix through a balls-and-bins abstraction. Our analysis revealed a sharp density threshold beyond which spatial ZF precoding provably collapses. To overcome this limitation, we proposed STAB, which exploits residual Doppler shifts as an additional degree of freedom, and developed the SDS algorithm for joint space-Doppler user scheduling. Both asymptotic analysis and finite-dimensional simulations confirmed that STAB with SDS achieves substantial sum rate gains over conventional methods in dense user regimes. {Future work could extend this framework to multi-satellite systems, addressing asynchronous arrivals due to propagation delays and Doppler ambiguities arising when STAB uses Doppler for user discrimination.}

\appendices

\section{Proof of Lemma~\ref{lem:max_load}} \label{ap:A}
Let $n_b$ denote the load of bin $b$, $b=1,\ldots,B$, and define
\begin{align}
n_{\max} = \max_{1\le b\le B} n_b, \quad \lambda_M = \frac{K}{B} = M^{p+r-1}.
\end{align}
We treat the three loading regimes separately.

\textbf{1) Sparse regime} $(p<1-r)$:
In this case, $\lambda_M \to 0$ as $M \to \infty$. For any fixed integer $t\ge 1$, the probability that $n_{\max}$ exceeds $t$ is bounded as
\begin{align}
\mathbb{P}(n_{\max}\ge t) &\le B \cdot \mathbb{P}\!\left(\mathcal{B}\!\left(K,\frac{1}{B}\right)\ge t\right) \notag
\\&\le B \binom{K}{t} \left(\frac{1}{B} \right)^t \le B \left(\frac{e\lambda_M}{t}\right)^t,
\end{align}
where $\mathcal{B}(\cdot)$ denotes a binomial random variable. The first and second inequalities follow from the union bound, and the third is established by $\binom{K}{t} \le \left(\frac{eK}{t}\right)^t$.
Since
\begin{align}
B\left(\frac{e\lambda_M}{t}\right)^t
&=
M^{\,1-r+t(p+r-1)}\left(\frac{e}{t}\right)^t,
\end{align}
and $p+r-1<0$, one can choose a sufficiently large $t$ such that
$1-r+t(p+r-1)<0$. Hence $\mathbb{P}(n_{\max}\ge t)\to 0$ as $M \to \infty$, which implies
\begin{align}
n_{\max}=O(1) \quad \text{w.h.p.}
\end{align}

\textbf{2) Critical regime} $(p=1-r)$:
In this case, $\lambda_M = \frac{K}{B} = 1$, so that $K=B$. Let
\begin{align}
k_\alpha = \alpha \frac{\log B}{\log\log B},
\end{align}
and define, for each bin $b=1,\ldots,B$, $X_b = \mathbf{1}\{n_b \ge k_\alpha\}$, and the counting variable $X = \sum_{b=1}^B X_b$.
Then $X>0$ if and only if $n_{\max}\ge k_\alpha$. Since the bins are identically distributed,
\begin{align}
\mathbb{E}[X]
= B\,\mathbb{P}\!\left(\mathcal{B}\!\left(B,\frac{1}{B}\right)\ge k_\alpha\right).
\end{align}
Using that this upper tail is asymptotically dominated by $\mathbb{P}\!\left(\mathcal{B}\!\left(B,\frac{1}{B}\right) = k_\alpha\right)$ and applying
Stirling's formula, one obtains
\begin{align}
\mathbb{E}[X] = B^{\,1-\alpha+o(1)}.
\end{align}
Hence, if $\alpha>1$, then $\mathbb{E}[X]\to 0$, and Markov's inequality yields $\mathbb{P}(X>0)=o(1)$.
Therefore,
\begin{align}
\mathbb{P}(n_{\max}\ge k_\alpha)=o(1),
\end{align}
which gives the upper bound.

On the other hand, if $0<\alpha<1$, then $\mathbb{E}[X]\to\infty$. Moreover, for
$b\neq b'$,
\begin{align}
\mathbb{E}[X_bX_{b'}]
&= \mathbb{P}(n_b\ge k_\alpha,\;n_{b'}\ge k_\alpha) \notag
\\
&\le (1+o(1))\mathbb{P}(n_b\ge k_\alpha)^2.
\end{align}
Thus,
\begin{align}
\mathbb{E}[X^2]
= \sum_{b=1}^B \mathbb{E}[X_b]
+ \sum_{b\ne b'} \mathbb{E}[X_bX_{b'}]
= (1+o(1))\mathbb{E}[X]^2.
\end{align}
By the second moment method,
\begin{align}
\mathbb{P}(X>0)\ge \frac{\mathbb{E}[X]^2}{\mathbb{E}[X^2]}= 1-o(1).
\end{align}
Hence,
\begin{align}
\mathbb{P}(n_{\max}\ge k_\alpha)=1-o(1),
\end{align}
which gives the lower bound.

Combining the upper and lower bounds and substituting $B=M^{1-r}$, we conclude that
\begin{align}
n_{\max} = \Theta \left( \frac{\log M}{\log\log M} \right) \qquad \text{w.h.p.}
\end{align}
A detailed proof can be found in \cite{raab:balls:1998}.

\textbf{3) Dense regime} $(p>1-r)$: Here $\lambda_M = M^{p+r-1}\to\infty$. The lower
bound follows directly from the pigeonhole principle as
\begin{align}
n_{\max}\ge \left\lceil \frac{K}{B} \right\rceil \ge M^{p+r-1}.
\end{align}
For the upper bound, note that each $n_b\sim \mathcal{B}(K,1/B)$ has mean
$\lambda_M$. By Chernoff's inequality,
\begin{align}
\mathbb{P}(n_b \ge 2\lambda_M) \le e^{-\lambda_M/3}.
\end{align}
Applying the union bound over all $B$ bins yields
\begin{align}
\mathbb{P}(n_{\max}\ge 2\lambda_M) &\le B e^{-\lambda_M/3} \notag
\\&=M^{1-r} e^{-M^{p+r-1}/3}.
\end{align}
Hence, $\mathbb{P}(n_{\max} \ge 2\lambda_M) \to 0$ as $M \to \infty$. This implies that
\begin{align}
n_{\max} &= \Theta(\lambda_M) = \Theta\!\left(M^{p+r-1}\right) \quad \text{w.h.p.}
\end{align}

\section{Proof of Lemma~\ref{lem:lambda_min_ub}} \label{ap:B}
{
We first consider $D=1$. Let $u_0$ be the left endpoint of the occupied resolution bin and write $u_k=u_0+\vartheta_k/M$, where $\vartheta_k\in[0,1]$. For any $\mathbf w\in\mathbb C^n\setminus\{\mathbf 0\}$, the Rayleigh quotient gives
\begin{equation}
\lambda_1(\mathbf G)
\le
\frac{\mathbf w^{\mathsf H}\mathbf G\mathbf w}{\|\mathbf w\|_2^2}
=
\frac{1}{M\|\mathbf w\|_2^2}
\sum_{m=0}^{M-1}
\left|
\sum_{k=1}^n
w_k e^{j2\pi(m/M)\vartheta_k}
\right|^2.
\label{eq:ltrB_rayleigh}
\end{equation}
The common phase $e^{j2\pi m u_0}$ does not appear because it has unit magnitude. Choose a nonzero $\mathbf w$ satisfying
\begin{equation}
\sum_{k=1}^n w_k\vartheta_k^\ell=0,
\qquad
\ell=0,\ldots,J_1(n)-1.
\label{eq:ltrB_moment}
\end{equation}
Since $J_1(n)=n-1$, these are $n-1$ homogeneous linear constraints on $n$ coefficients, so such a vector always exists. The constraints cancel all Taylor terms of order below $J_1(n)$. Hence, uniformly over $m=0,\ldots,M-1$,
\begin{align}
\left|
\sum_{k=1}^n
w_k e^{j2\pi(m/M)\vartheta_k}
\right|
&\le
\|\mathbf w\|_1
\sum_{\ell=J_1(n)}^\infty
\frac{(2\pi)^\ell}{\ell!}\notag\\
&\le
\|\mathbf w\|_1 e^{2\pi}
\frac{(2\pi)^{J_1(n)}}{J_1(n)!}.
\label{eq:ltrB_tail}
\end{align}
Substituting \eqref{eq:ltrB_tail} into \eqref{eq:ltrB_rayleigh}, using $\|\mathbf w\|_1\le\sqrt n\|\mathbf w\|_2$, and applying $J_1(n)!\ge[J_1(n)/e]^{J_1(n)}$ yield
\begin{align}
\lambda_1(\mathbf G)
&\le
n e^{4\pi}
\left(
\frac{2\pi e}{J_1(n)}
\right)^{2J_1(n)}\notag\\
&=
n e^{4\pi}
\left(
\frac{2\pi e}{n-1}
\right)^{2(n-1)},
\label{eq:ltrB_final}
\end{align}
which proves \eqref{eq:lambda_min_ub_1d}.

For a $D$-dimensional resolution bin, let $\boldsymbol\vartheta_k\in[0,1]^D$ denote the normalized in-bin offset of user $k$. Choose $\mathbf w\neq\mathbf 0$ such that
\begin{equation}
\sum_{k=1}^n
w_k\boldsymbol\vartheta_k^{\boldsymbol\ell}=0,
\qquad
|\boldsymbol\ell|<J_D(n),
\label{eq:ltrB_moment_d}
\end{equation}
where $\boldsymbol\ell=(\ell_1,\ldots,\ell_D)$ is a nonnegative multi-index and $|\boldsymbol\ell|=\ell_1+\cdots+\ell_D$. The number of constraints is
$\binom{J_D(n)-1+D}{D}\le n-1$ by \eqref{eq:taylor_order}, so a nonzero solution exists. These constraints cancel every multivariate Taylor term of total degree below $J_D(n)$. Since the inner product between a normalized array index in $[0,1)^D$ and $\boldsymbol\vartheta_k$ is at most $D$, the same tail argument gives
\begin{equation}
\lambda_1(\mathbf G)
\le
n e^{4\pi D}
\left(
\frac{2\pi D e}{J_D(n)}
\right)^{2J_D(n)},
\label{eq:ltrB_final_d}
\end{equation}
which proves \eqref{eq:lambda_min_ub}.
}

\section{Proof of Lemma~\ref{lem:zf_sumrate}} \label{ap:C}
Let $\mathbf{G}_S \in \mathbb{C}^{n \times n} $ denote the principal submatrix of $\mathbf{G}$ corresponding to the $n$ users in a most-loaded bin, i.e., $n_{\max} = n$. Then,
\begin{subequations}\label{eq:lemma2_upper}
\begin{align}
\mathbb{E}[R_\Sigma  \mid n_{\max} = n]
&{\le}\;
K\log_2 \big(1+\rho M\,\lambda_1(\mathbf{G})\big)
\label{eq:lemma2_upper_a}\\
&{\le}\;
M^{p}\log_2 \big(1+\rho M\,\lambda_1(\mathbf{G}_S)\big)
\label{eq:lemma2_upper_b}\\
&{\le}\;
{M^{p}\log_2 \left(1+\rho M\,n\,e^{4\pi}\left(\tfrac{2\pi e}{n-1}\right)^{2(n-1)}\right)}
\label{eq:lemma2_upper_c}
\end{align}
\end{subequations}
where \eqref{eq:lemma2_upper_a} follows from \eqref{eq:sumrate_lambda_min_ub}, and \eqref{eq:lemma2_upper_b} follows from the Cauchy interlacing theorem for Hermitian matrices, which yields
$\lambda_1(\mathbf{G})\le \lambda_1(\mathbf{G}_S)$. {Finally, \eqref{eq:lemma2_upper_c} follows from the one-dimensional bound \eqref{eq:lambda_min_ub_1d} of Lemma~\ref{lem:lambda_min_ub} applied to the most-loaded bin.}

\section{Proof of Theorem~\ref{thm:ERsum_ULA}} \label{ap:D}
We partition the spatial support $\CMcal{U} = [-\frac{R}{2H},\frac{R}{2H}]$ into bins of width $1/M$. The total number of bins is given by
\begin{equation}
B = \frac{|\CMcal{U}|}{1/M} = \frac{R/H}{1/M} = \frac{MR}{H} = M^{1-r}.
\end{equation}

Due to the small-angle approximation, the system can be modeled as a balls-and-bins problem where $K = M^p$ users are placed i.i.d. uniformly into $B = M^{1-r}$ bins. To derive the upper scaling law for $\mathbb{E}[R_\Sigma]$, we apply the law of total expectation as
\begin{equation}
\mathbb{E}[R_\Sigma] = \mathbb{E}[R_\Sigma|\CMcal{B}]\mathbb{P}(\CMcal{B}) + \mathbb{E}[R_\Sigma|\CMcal{B}^c]\mathbb{P}(\CMcal{B}^c),
\end{equation}
where $\CMcal{B}$ denotes the event that the maximum load $n_{\max}$ satisfies the relevant bounds, an event that holds w.h.p. as $M \to \infty$. To bound the contribution of the event $\CMcal{B}^c$, we note that the sum rate is always upper bounded by
\begin{align}
    R_\Sigma \leq K \log_2 \left( 1 + \frac{\rho M}{K} \right),
\end{align}
since $\operatorname{tr}(\mathbf{G}^{-1}) = \sum_{i=1}^K \lambda_i(\mathbf{G})^{-1} \geq K$. Consequently, the average sum rate can be bounded as
\begin{equation}
\mathbb{E}[R_\Sigma] \leq \mathbb{E}[R_\Sigma|\CMcal{B}]\mathbb{P}(\CMcal{B}) + K \log_2 \left( 1 + \frac{\rho M}{K} \right) \mathbb{P}(\CMcal{B}^c). \label{eq:law_of_total}
\end{equation}

The scaling behavior of the maximum load $n_{\max}$ is given by Lemma~\ref{lem:max_load}. Accordingly, we analyze the average sum rate under the three scaling regimes defined therein.

{Throughout this appendix, Lemma~\ref{lem:lambda_min_ub} is used in asymptotic form. For fixed $D$,
\begin{equation}
J_D(n)=(D!\,n)^{1/D}(1+o(1)),
\end{equation}
and therefore
\begin{align}
ne^{4\pi D}
\left(\frac{2\pi D e}{J_D(n)}\right)^{2J_D(n)} =
\exp\!\left[-\left(\frac{2(D!)^{1/D}}{D}+o(1)\right)n^{1/D}\log n\right].
\label{eq:lambda_asymp}
\end{align}
}

\textbf{1) Sparse regime} $(p<1-r)$:
By Lemma~\ref{lem:max_load}, there exists a constant $C>0$ such that
$\CMcal{B}_s = \{n_{\max}\le C\}$ holds w.h.p. On $\CMcal{B}_s$, Lemma~\ref{lem:zf_sumrate} gives
\begin{align}
\mathbb{E}[R_\Sigma|\CMcal{B}_s]
&\lesssim M^p \log_2\!\left(1+\rho M\right)
\lesssim M^p\log M.
\end{align}
On $\CMcal{B}_s^c$, the deterministic bound above is also $O(M^p\log M)$. Therefore, from \eqref{eq:law_of_total}, it follows that
\begin{align}
\mathbb{E}[R_\Sigma] \lesssim M^p\log M.
\end{align}

\textbf{2) Critical regime} $(p=1-r)$:
In this regime, we have $K/B = M^{p+r-1} = 1$ and $K = B = M^{1-r}$. Based on the results of Lemma~\ref{lem:max_load}, we know that the maximum load is concentrated on the scale of $\frac{\log B}{\log \log B}$. Accordingly, we define the predicted maximum load as $n_B = \frac{\log B}{\log \log B}$ and the margin as $\epsilon_B = \frac{1}{\log \log B}$. Let $t_B = \lfloor(1-\epsilon_B)n_B\rfloor$ denote the threshold, which exactly corresponds to $k_\alpha$ in the proof of Lemma~\ref{lem:max_load} with $\alpha = 1-\epsilon_B$. To utilize the result of \eqref{eq:law_of_total}, we define the event $\CMcal{B}_c = \{n_{\max} \ge t_B\}$. We first show that the second term of the upper bound in \eqref{eq:law_of_total} converges to zero as $M \to \infty$ and then derive the result for the first term using Lemma~\ref{lem:zf_sumrate}.

To bound the probability $\mathbb{P}(\CMcal{B}_c^c)$, we introduce a Poissonization argument. To be specific, we assume that the total number of users follows a Poisson distribution $K_e \sim \mathrm{Pois}(K)$ where $\mathrm{Pois}(\cdot)$ denotes a Poisson random variable. When $K_e$ users are uniformly distributed across $B$ bins, the number of users in each bin $b=1,\dots,B$ becomes an i.i.d. random variable $n_b \sim \mathrm{Pois}(1)$. Consequently, the Poissonized probability $\mathbb{P}_\mathrm{Pois}(\CMcal{B}_c^c)$ factorizes as
\begin{align}
\mathbb{P}_\mathrm{Pois}(\CMcal{B}_c^c) = (1-p_B)^B \le \exp(-Bp_B),
\end{align}
where $p_B = \mathbb{P}(n_1 \ge t_B)$. Next, by using the lower bound $p_B \ge \mathbb{P}(n_1=t_B) = \frac{e^{-1}}{t_B!}$ and Stirling's formula to calculate the factorial term, we obtain
\begin{align}
\log p_B \ge -t_B\log t_B + O(t_B) = -(1-\epsilon_B)\log B + o(\log B).
\end{align}
This implies that $p_B \ge B^{-(1-\epsilon_B)+o(1)}$, which leads to $Bp_B \ge B^{\epsilon_B+o(1)}$. Thus, the Poissonized probability decays as
\begin{align}
\mathbb{P}_\mathrm{Pois}(\CMcal{B}_c^c) \le \exp\!\left(-B^{\epsilon_B+o(1)}\right). \label{eq:pois_upper}
\end{align}
By returning to the original model with a fixed $K$ through conditional probability, we have
\begin{align}
\mathbb{P}(\CMcal{B}_c^c) = \mathbb{P}_\mathrm{Pois}(\CMcal{B}_c^c \mid K_e=K) \le \frac{\mathbb{P}_\mathrm{Pois}(\CMcal{B}_c^c)}{\mathbb{P}(K_e=K)}.
\end{align}
Since $\mathbb{P}(K_e=K) = \Theta(K^{-1/2})$ via Stirling's formula, this value is negligible compared to the upper bound of $\mathbb{P}_\mathrm{Pois}(\CMcal{B}_c^c)$ in \eqref{eq:pois_upper}, which results in $\mathbb{P}(\CMcal{B}_c^c) \le \exp(-B^{\epsilon_B+o(1)})$. Therefore, the contribution of $\CMcal{B}_c^c$ to the total expectation vanishes asymptotically as
\begin{align}
K\log_2\!\left(1+\frac{\rho M}{K}\right)\mathbb{P}(\CMcal{B}_c^c) \to 0.
\end{align}

Conditioned on the event $\CMcal{B}_c$, the maximum load satisfies $n_{\max} \ge t_B$. Because the sum rate upper bound in Lemma~\ref{lem:zf_sumrate} is monotonically decreasing for sufficiently large $n$, {applying \eqref{eq:lambda_asymp}} we obtain
\begin{align}
R_\Sigma \lesssim K\log_2\!\left( 1+\rho M\,{e^{-(2+o(1))t_B\log t_B}} \right).
\end{align}
Analyzing the exponent reveals that
\begin{align}
{(2+o(1))\,t_B\log t_B = (2+o(1))\log B .}
\end{align}
From this, we derive the relation {$e^{-(2+o(1))t_B\log t_B} = B^{-2+o(1)}$}. By substituting $B=M^{1-r}$, the scale of the SNR term becomes
\begin{align}
\rho M\,{e^{-(2+o(1))t_B\log t_B}} = \rho M\,B^{-2+o(1)} = \rho M^{2r-1+o(1)}.
\end{align}
Thus, the conditional expectation is upper bounded as $\mathbb{E}[R_\Sigma \mid \CMcal{B}_c] \lesssim M^{1-r}\log_2(1+\rho M^{2r-1+o(1)})$. By substituting these results into \eqref{eq:law_of_total} and evaluating the asymptotic behavior according to the value of $r$, we arrive at the following conclusion
\begin{align}
\mathbb{E}[R_\Sigma] \lesssim \begin{cases} M^{r+o(1)}, & r<\frac{1}{2}, \\M^{1/2+o(1)}, & r=\frac{1}{2}, \\M^{1-r+o(1)}\log M, & r>\frac{1}{2}. \end{cases}
\end{align}

\textbf{3) Dense regime} $(p>1-r)$:
By the pigeonhole principle,
\begin{align}
n_{\max}
\ge \left\lceil \frac{K}{B}\right\rceil
= \left\lceil M^{p+r-1}\right\rceil.
\end{align}
Since $p+r-1>0$, the lower bound diverges. Applying Lemma~\ref{lem:zf_sumrate} {through \eqref{eq:lambda_asymp}},
\begin{align}
R_\Sigma
\lesssim M^p\log_2\!\left(1+\rho M\,{e^{-(2+o(1))n_{\max}\log n_{\max}}}\right)
\to 0,
\end{align}
and hence $\mathbb{E}[R_\Sigma]\to 0.$

{\section{Proof of Corollaries~\ref{cor:separation} and~\ref{cor:separation_upa}}
\label{ap:E}
Set $(D,\chi)=(2,1-r)$ for the ULA and $(D,\chi)=(3,1-2r)$ for the UPA.
Since $q=p-\chi+\varepsilon$, the joint space-Doppler domain contains $N_{\mathrm J}=M^{\chi+q}=M^{p+\varepsilon}$
resolution bins. Here, $N_{\mathrm J}$ governs the random user
geometry, whereas $ML$ is the dimension of each space-time channel
vector and therefore enters the beam energy and SINR. Although
$K/N_{\mathrm J}=M^{-\varepsilon}\to0$, the $K^2$ possible user
pairs mean that exceptional close pairs may still occur. The graph
construction therefore isolates this local crowding rather than
requiring the full Gram matrix to be diagonally dominant.

Normalize each spatial coordinate by its array resolution and the Doppler coordinate by $L$, and denote the resulting max-distance between users $i$ and $j$ by $\Delta_{ij}$, with wrap-around in Doppler. Connect users whenever $\Delta_{ij}\le\Delta_{\mathrm c}$, where
\begin{equation}
\Delta_{\mathrm c}=M^{\varepsilon/(2D)},
\label{eq:ach_cluster_radius}
\end{equation}
and call each connected component a cluster.

Let $\mathcal N_t$ denote the number of connected $t$-user groups. Since every such group contains a spanning tree, the standard tree-counting bound for random geometric graphs \cite{penrose:rgg:03} gives
\begin{equation}
\mathbb E[\mathcal N_t]
\lesssim
K\left(\frac{K\Delta_{\mathrm c}^D}{N_{\mathrm J}}\right)^{t-1}
=
M^{p-\varepsilon(t-1)/2}.
\label{eq:ach_cluster_count}
\end{equation}
Hence, with
\begin{equation}
n_0=\left\lceil\frac{2p}{\varepsilon}\right\rceil+1,
\end{equation}
every cluster contains at most $n_0$ users with probability $1-o(1)$.

Define the minimum separation scale
\begin{equation}
\Delta_{\min}=M^{-(p-3\varepsilon/4)/D}.
\label{eq:ach_min_scale}
\end{equation}
A union bound over all user pairs yields
\begin{equation}
\mathbb P\!\left\{
\min_{i\ne j}\Delta_{ij}<\Delta_{\min}
\right\}
\lesssim
\frac{K^2\Delta_{\min}^D}{N_{\mathrm J}}
=
M^{-\varepsilon/4}.
\label{eq:ach_min_distance}
\end{equation}
Let $\CMcal{E}_{\mathrm g}$ be the event on which the cluster size and
minimum separation properties above hold. On $\CMcal{E}_{\mathrm g}$,
distinct clusters are separated by more than $\Delta_{\mathrm c}$,
each cluster has diameter $O(n_0\Delta_{\mathrm c})$, and
$\mathbb P(\CMcal{E}_{\mathrm g})=1-o(1)$. These two properties are
complementary: bounded cluster size leaves only a fixed number of
local interpolation constraints, while $\Delta_{\min}$ limits their
conditioning cost to a fixed power of $M$. The separation
$\Delta_{\mathrm c}$ is then used to suppress inter-cluster responses.

For a cluster $\CMcal{C}$ and $i\in\CMcal{C}$, define the local amplification factor
\begin{equation}
\Upsilon_i
=
\prod_{j\in\CMcal{C}\setminus\{i\}}
\max\!\left\{
\frac{\Delta_{\mathrm c}}{\Delta_{ij}},1
\right\}.
\label{eq:ach_local_cost}
\end{equation}
The construction can be viewed as a tapered trigonometric Lagrange
interpolant: for user $i$, it imposes unit response at $i$ and zeros
at the other nodes of $\CMcal{C}$. The factor $\Upsilon_i$ measures
the amplification needed to impose these zeros when cluster members
are close. Since the taper is included in the interpolation
construction, the local constraints remain exact while the response
decays polynomially away from $\CMcal{C}$.

The localized Lagrange construction of \cite{kunis:laa:20}, combined with a fixed taper of sufficiently large order $N_{\mathrm{tap}}$, provides a local space-time precoding vector $\bar{\mathbf f}_i^{(0)}$ satisfying
\begin{equation}
\bar{\mathbf h}_j^{\mathsf H}\bar{\mathbf f}_i^{(0)}
=
\mathbf{1}\{i=j\},
\qquad j\in\CMcal{C},
\label{eq:ach_local_interp}
\end{equation}
\begin{equation}
\left|
\bar{\mathbf h}_j^{\mathsf H}\bar{\mathbf f}_i^{(0)}
\right|
\lesssim
\Upsilon_i(1+\Delta_{ij})^{-N_{\mathrm{tap}}},
\qquad j\notin\CMcal{C},
\label{eq:ach_local_decay}
\end{equation}
and
\begin{equation}
ML\|\bar{\mathbf f}_i^{(0)}\|_2^2
\lesssim
\Upsilon_i^2.
\label{eq:ach_local_energy}
\end{equation}
On $\CMcal{E}_{\mathrm g}$,
$\Upsilon_i\le(\Delta_{\mathrm c}/\Delta_{\min})^{n_0-1}=M^{O(1)}$.
Thus, a sufficiently large $N_{\mathrm{tap}}$, fixed independently of
$M$, can dominate this polynomial amplification beyond
$\Delta_{\mathrm c}$, making every inter-cluster leakage coefficient
at most $M^{-(p+2)}$.

Collecting the local vectors in $\bar{\mathbf F}_0$ gives
\begin{equation}
\bar{\mathbf H}^{\mathsf H}\bar{\mathbf F}_0
=
\mathbf I_K+\boldsymbol\Phi.
\label{eq:ach_local_matrix}
\end{equation}
Since $\bar{\mathbf F}_0$ is exact within every cluster,
$\boldsymbol\Phi$ contains only inter-cluster responses. Each such
entry is at most $M^{-(p+2)}$, and at most $K=M^p$ entries occur in
any row or column. Hence,
\begin{equation}
\|\boldsymbol\Phi\|_2
\le
\sqrt{\|\boldsymbol\Phi\|_1\|\boldsymbol\Phi\|_\infty}
=
O(M^{-2}).
\label{eq:ach_leakage}
\end{equation}
Thus, the corrected right inverse
\begin{equation}
\bar{\mathbf F}_{\mathrm{ri}}
=
\bar{\mathbf F}_0(\mathbf I_K+\boldsymbol\Phi)^{-1}
\label{eq:ach_corrected_beam}
\end{equation}
satisfies $\bar{\mathbf H}^{\mathsf H}\bar{\mathbf F}_{\mathrm{ri}}=\mathbf I_K$ and
$\|\bar{\mathbf F}_{\mathrm{ri}}\|_F^2\le(1+o(1))\|\bar{\mathbf F}_0\|_F^2$.

It remains to show that rare close configurations do not dominate
the total beam power. Fix an anchor user and sum over the possible
additional cluster members. A member at distance
$x<\Delta_{\mathrm c}$ contributes the squared amplification factor
$(\Delta_{\mathrm c}/x)^2$, while the probability of lying in the
corresponding $D$-dimensional shell is
$O(x^{D-1}dx/N_{\mathrm J})$. At larger distances, the amplification
factor is one, producing the volume term below. Therefore, the
expected multiplicative energy factor associated with one additional
cluster member satisfies
\begin{equation}
\mathcal I_M
\lesssim
\frac{K}{N_{\mathrm J}}
\left[
\Delta_{\mathrm c}^2
\int_{\Delta_{\min}}^{\Delta_{\mathrm c}}x^{D-3}dx
+
\Delta_{\mathrm c}^D
\right]
=
M^{-\varepsilon/2+o(1)}.
\label{eq:ach_energy_factor}
\end{equation}
The lower-limit singularity is logarithmic only for $D=2$ and is
absent for $D=3$. Thus, $\mathcal I_M\to0$. Since the cluster size is
fixed, every additional member contributes another factor
$\mathcal I_M$; multiuser clusters are therefore lower order, while
singleton clusters determine the $M^{p+o(1)}$ total power order.
Accordingly,
\begin{equation}
\mathbb E\!\left[
ML\|\bar{\mathbf F}_0\|_F^2\mathbf{1}_{\CMcal{E}_{\mathrm g}}
\right]
\lesssim
K\sum_{m=1}^{n_0}m\mathcal I_M^{m-1}
=
M^{p+o(1)}.
\label{eq:ach_power_mean}
\end{equation}
Markov's inequality and the leakage correction imply
\begin{equation}
ML\|\bar{\mathbf F}_{\mathrm{ri}}\|_F^2
\le
M^{p+\varepsilon/2+o(1)}
\label{eq:ach_power_whp}
\end{equation}
with probability $1-o(1)$.

The final step uses the variational characterization of minimum norm
ZF: it cannot use more energy than any feasible right inverse.
Therefore, for any right inverse satisfying
$\bar{\mathbf H}^{\mathsf H}\bar{\mathbf F}_{\mathrm{ri}}=\mathbf I_K$,
\begin{equation}
\operatorname{tr}(\bar{\mathbf G}^{-1})
\le
ML\|\bar{\mathbf F}_{\mathrm{ri}}\|_F^2.
\label{eq:ach_inverse_gram_bound}
\end{equation}
Consequently,
\begin{equation}
\overline{\mathrm{SINR}}
\ge
\rho M^{1-\chi+\varepsilon/2-o(1)}
\label{eq:ach_sinr}
\end{equation}
with probability $1-o(1)$. Since $K/L=M^{\chi-\varepsilon}$ and
the SINR in \eqref{eq:ach_sinr} grows polynomially with $M$, the sum
rate satisfies
\begin{equation}
\bar R_\Sigma
=
\Omega\!\left(
M^{\chi-\varepsilon}\log M
\right)
\qquad \text{w.h.p.}
\label{eq:ach_rate_general}
\end{equation}
Since \eqref{eq:ach_rate_general} holds with probability $1-o(1)$
and $\bar R_\Sigma\ge0$,
\begin{equation}
\mathbb E[\bar R_\Sigma]
=
\Omega\!\left(
M^{\chi-\varepsilon}\log M
\right).
\label{eq:ach_rate_expectation}
\end{equation}

Finally, substituting $\chi=1-r$ for the ULA and
$\chi=1-2r$ for the UPA, and combining these lower bounds with
Theorems~\ref{thm:ERsum_STAB} and~\ref{thm:ERsum_UPA_STAB}, respectively, yields
\begin{equation}
\mathbb E[\bar R_\Sigma]
=
\begin{cases}
\Theta\!\left(
M^{1-r-\varepsilon}\log M
\right), & \text{ULA},\\[5pt]
\Theta\!\left(
M^{1-2r-\varepsilon}\log M
\right), & \text{UPA}.
\end{cases}
\label{eq:ach_rate_cases}
\end{equation}}

\bibliographystyle{IEEEtran}
\bibliography{ref_stab}

\end{document}